\documentclass[12pt,a4paper]{amsart}

\usepackage{amsmath,amsfonts,amssymb,mathtools,extarrows}

\usepackage[hmargin=2cm,vmargin=2cm]{geometry}

\usepackage{hyperref}
\usepackage{nameref,zref-xr}     

\usepackage{enumerate}               

\newcommand{\mbZ}{\mathbb Z}
\newcommand{\mbC}{\mathbb C}

\newcommand{\oM}{\overline{\mathcal M}}

\newcommand{\tu}{{\widetilde u}}
\newcommand{\og}{\overline g}
\newcommand{\oh}{\overline h}
\newcommand{\hLambda}{\widehat\Lambda}
\newcommand{\cT}{\mathcal T}
\def\cM{{\mathcal{M}}}
\def\oM{{\overline{\mathcal{M}}}}
\def\CP{{{\mathbb C}{\mathbb P}}}
\renewcommand{\Im}{\mathrm{Im}}

\def\d{{\partial}}

\renewcommand{\>}{\right>}
\newcommand{\eps}{\varepsilon}

\newcommand{\cA}{\mathcal A}
\newcommand{\hcA}{\widehat{\mathcal A}}
\newcommand{\DR}{\mathrm{DR}}

\newcommand{\even}{\mathrm{even}}
\newcommand{\ct}{\mathrm{ct}}

\DeclareMathOperator{\Deg}{Deg}

\newcommand{\gl}{\mathrm{gl}}

\newcommand{\hu}{\widehat{u}}

\newcommand{\hOmega}{\widehat{\Omega}}
\newcommand{\hE}{\widehat{E}}
\newcommand{\diag}{\mathrm{diag}}
\newcommand{\un}{{1\!\! 1}}
\newcommand{\mcF}{\mathcal{F}}

\newcommand{\of}{\overline{f}}
\newcommand{\oQ}{{\overline{Q}}}
\newcommand{\oP}{{\overline{P}}}
\newcommand{\tK}{\widetilde{K}}
\newcommand{\oR}{{\overline{R}}}

\newcommand{\orig}{\mathrm{orig}}
\newcommand{\otu}{\overline{\widetilde{u}}}
\newcommand{\triv}{\mathrm{triv}}
\newcommand{\KdV}{\mathrm{KdV}}
\newcommand{\rspin}{\text{$r$-spin}}
\newcommand{\pt}{\mathrm{pt}}
\newcommand{\GD}{\mathrm{GD}}
\DeclareMathOperator{\res}{res}
\newcommand{\tX}{\widetilde{X}}
\newcommand{\tS}{\widetilde{S}}
\newcommand{\Td}{\mathrm{Td}}

\newtheorem{theorem}{Theorem}[section]
\newtheorem{proposition}[theorem]{Proposition}
\newtheorem{lemma}[theorem]{Lemma}

\newtheorem{conjecture}[theorem]{Conjecture}

\theoremstyle{remark}

\newtheorem{remark}[theorem]{Remark}
\newtheorem{example}[theorem]{Example}

\theoremstyle{definition}

\newtheorem{definition}[theorem]{Definition}
\newtheorem{notation}[theorem]{Notation}

\usepackage{color}

\numberwithin{equation}{section}

\title{Towards a bihamiltonian structure for the double ramification hierarchy}

\author{Alexandr Buryak}
\address{A.~Buryak:\newline 
Faculty of Mathematics, National Research University Higher School of Economics, \newline
6 Usacheva str., 119048 Moscow, Russian Federation; and \newline
Faculty of Mechanics and Mathematics, Lomonosov Moscow State University, \newline 
GSP-1, 119991 Moscow, Russian Federation}
\email{aburyak@hse.ru}

\author{Paolo Rossi}
\address{P.~Rossi:\newline
Dipartimento di Matematica ``Tullio Levi-Civita'', Universit\`a degli Studi di Padova,\newline
Via Trieste 63, 35121 Padova, Italy}
\email{paolo.rossi@math.unipd.it}

\author{Sergey Shadrin}
\address{S.~Shadrin:\newline 
Korteweg--de Vries Instituut voor Wiskunde, Universiteit van Amsterdam,\newline
Postbus 94248, 1090GE Amsterdam, The Netherlands}
\email{s.shadrin@uva.nl}

\dedicatory{To the memory of Boris Dubrovin, our teacher and friend}

\begin{document}

\begin{abstract}
We propose a remarkably simple and explicit conjectural formula for a bihamiltonian structure of the double ramification hierarchy corresponding to an arbitrary homogeneous cohomological field theory. Various checks are presented to support the conjecture.
\end{abstract}

\date{\today}

\maketitle

\tableofcontents

\section*{Introduction}

Cohomological field theories (or CohFTs for brevity) are systems of cohomology classes on the moduli space $\oM_{g,n}$ of stable algebraic curves of genus $g$ with $n$ marked points. They were introduced by Kontsevich and Manin in~\cite{KM94} to axiomatize the properties of Gromov--Witten classes of a given target variety. Their compatibility with the strata structure and natural morphisms between moduli spaces makes them powerful tools for probing the cohomology of~$\oM_{g,n}$ and its tautological ring in particular.\\

Since the Kontsevich--Witten theorem \cite{Wit91,Kon92} stating that the generating series of integrals over $\oM_{g,n}$ of monomials in psi classes (the Chern classes of tautological line bundles) is the logarithm of the tau function of a special solution to the Korteweg--de Vries (KdV) hierarchy, it is well known that integrable hierarchies of evolutionary, Hamiltonian, tau-symmetric PDEs control the intersection theory of CohFTs.\\

Dubrovin and Zhang \cite{DZ01} give a systematic construction of such integrable hierarchy starting from a semisimple cohomological field theory. Their framework gives, among other things, the language for stating the analogue of the Kontsevich--Witten theorem for any semisimple CohFT, where the KdV hierarchy is replaced by the relevant Dubrovin--Zhang (DZ) hierarchy.\\

In fact, Dubrovin and Zhang's method postulates the existence of a bihamiltonian structure (a pair of compatible Poisson structure producing the flows by a recursive procedure) for the DZ hierarchy of a homogeneous semisimple CohFTs. While one of the two Poisson structures was proved to exist (in the differential polynomial class) in \cite{BPS12a,BPS12b}, the existence of the second Hamiltonian structure is still an open problem. Notice that the approach of \cite{BPS12a,BPS12b} also clarifies how the construction of the DZ hierarchy (as a Hamiltonian system with respect to the first Poisson bracket) can be based on the axioms and properties of CohFTs together with semisimplicity, without requiring the existence of the second Poisson structure or homogeneity. Nonetheless, the existence of this second Hamiltonian structure in the homogeneous case remains an important unproven feature of the DZ hierarchy.\\

In \cite{Bur15} a novel approach to constructing integrable hierarchies starting from CohFTs was introduced. This approach does not require semisimplicity of the CohFT and, although based again on the intersection theory on $\oM_{g,n}$, it employs different tautological classes, notably the double ramification cycle (an appropriate compactification of the locus of smooth curves whose marked points support a principal divisor), which explains why this hierarchy was called the double ramification (DR) hierarchy.\\

The DR hierarchy is always Hamiltonian with respect to a very simple Poisson structure, which, as opposed to the one for the DZ hierarchy, does not essentially depend on the underlying CohFT. The two hierarchies coincide by definition in the dispersionless (genus $0$) limit and, by a conjecture in~\cite{Bur15} called the {\it DR/DZ equivalence conjecture}, in the semisimple case they are related by a Miura transformation, which was completely identified in~\cite{BDGR18}. Although still unproven, the DR/DZ equivalence conjecture has accumulated a remarkable amount of evidence and verifications (see e.g. \cite{BG16,BDGR18,BDGR19,BGR19,DR19}).\\

In this paper we propose a very simple formula for a second Poisson structure and collect some evidence for it to give a bihamiltonian structure for the DR hierarchy. These Poisson brackets do depend on the homogeneous CohFT under consideration in a remarkably explicit way. We also compute the central invariants of the resulting bihamiltonian structure, finding that it is Miura equivalent to the conjectured DZ bihamiltonian structure. Finally, we confirm that in several well-known examples of CohFTs our formula does give the expected bihamiltonian structure of the corresponding hierarchy.\\

Notice that, assuming that our conjectural second Poisson bracket gives a bihamiltonian structure for the DR hierarchy, the existence of a bihamiltonian structure for the DZ hierarchy follows from the DR/DZ equivalence conjecture. In fact, this means that the Miura transformation mapping the DZ hierarchy to the DR hierarchy would simplify not only the first Poisson structure (making it virtually independent of the CohFT), but also the second one, for which no explicit formula was previously known.\\

\subsection*{Organization of the paper}
In Section \ref{section:main conjecture} we present our explicit formula and the main conjecture that it gives a bihamiltonian structure for the DR hierarchy of a homogeneous CohFT.\\

In Section \ref{section:general checks} we verify our conjecture in genus $0$. Then we prove that the level $0$ (primary) flows of the hierarchy can be obtained by bihamiltonian recursion from the level $-1$ integrals of motion (the Casimir functionals of the first Poisson bracket). Moreover, assuming our conjectural second Poisson bracket satisfies the Jacobi identity, we prove that it is compatible with the first one (i.e., their Schouten--Nijenhuis bracket vanishes).\\

In Section \ref{section:central invariants}, assuming again that our conjectural second Poisson bracket satisfies the Jacobi identity, we compute the central invariants of the Poisson pencil formed by the two Hamiltonian structures. Under mild hypotheses these central invariants classify bihamiltonian structures of the type relevant for the DR and DZ hierarchies up to Miura transformations. For the Dubrovin--Zhang bihamiltonian structure these invariants have been computed (see e.g.~\cite{Liu18}) and in this paper we show that they coincide with the ones for the DR bihamiltonian structure, as expected in light of the DR/DZ equivalence conjecture.\\

Finally, in Section \ref{section: examples} we prove our main conjecture for several important examples of {CohFTs}, including the trivial and $r$-spin CohFTs (for $r\leq 5$) and the Gromov--Witten theory of the projective line.\\

\subsection*{Notation and conventions} 
Throughout the text we use the Einstein summation convention for repeated upper and lower Greek indices.\\

When it doesn't lead to a confusion, we use the symbol $*$ to indicate any value, in the appropriate range, of a sub- or superscript.\\

For a topological space~$X$ let $H^*(X)$ denote the cohomology ring of~$X$ with the coefficients in $\mbC$.\\

\subsection*{Acknowledgements}

We are grateful to A.~Arsie and P.~Lorenzoni for valuable remarks about the preliminary version of the paper. We would like to thank G.~Carlet and F.~Hern\'andez Iglesias for useful discussions on closely related topics. We thank anonymous referees of our paper for valuable comments that allowed to improve the exposition of the paper.\\ 

The work of A.~B. (Sections 1 and 4) was supported by the grant no.~20-11-20214 of the Russian Science Foundation. S.~S. was supported by the Netherlands Organization for Scientific Research.\\


\section{Double ramification hierarchy and the main conjecture}\label{section:main conjecture}

In this section, after recalling the notion of cohomological field theory and the construction of the double ramification hierarchy, we present our conjectural formula for a bihamiltonian structure of the double ramification hierarchy.

\subsection{Cohomological field theories}\label{subsection:CohFT}

Let $\oM_{g,n}$ be the Deligne--Mumford moduli space of stable curves of genus $g$ with $n$ marked points, $g\geq 0$, $n\geq 0$, $2g-2+n>0$. Note that $\oM_{0,3}=\mathrm{pt}$, and throughout the text we silently use the identification $H^*(\oM_{0,3})\cong\mbC$. Recall the following system of standard maps between these spaces:
\begin{itemize}
\item $\pi_{g,n+1}\colon\oM_{g,n+1}\to\oM_{g,n}$ is the map that forgets the last marked point. 
\item $\gl_{g_1,I_1;g_2,I_2}\colon\oM_{g_1,n_1+1}\times\oM_{g_2,n_2+1}\to \oM_{g_1+g_2,n_1+n_2}$, is the gluing map that identifies the last marked points of curves of genus $g_1$ and $g_2$ and turns them into a node. The sets $I_1$ and $I_2$ of cardinality $n_1$ and $n_2$, $I_1\sqcup I_2 = \{1,\dots,n_1+n_2\}$, keep track of the relabelling of the remaining marked points. 
\item $\gl^{\mathrm{irr}}_{g,n+2}\colon\oM_{g,n+2}\to \oM_{g+1,n}$ is the gluing map that identifies the last two marked points and turns them into a node.
\end{itemize}
Abusing notation we denote these maps by $\pi$, $\gl_2$, and $\gl_1$, respectively. 

Let $V$ be a finite dimensional vector space of dimension $N$ with a distinguished vector $e\in V$, called the \emph{unit}, and a symmetric nondegenerate bilinear form $(\cdot,\cdot)$ on $V$, called the \emph{metric}. We fix a basis $e_1,\ldots,e_N$ in $V$ and let $(\eta_{\alpha\beta})$ denote the matrix of the metric in this basis, $\eta_{\alpha\beta}\coloneqq (e_\alpha,e_\beta)$, and $A_\alpha$ the coordinates of $e$ in this basis, $e=A^\alpha e_\alpha$. As usual, $\eta^{\alpha\beta}$ denotes the entries of the inverse matrix, $(\eta^{\alpha\beta})\coloneqq (\eta_{\alpha\beta})^{-1}$.

\begin{definition}[\cite{KM94}]
A \emph{cohomological field theory} (CohFT) is a system of linear maps 
$$
c_{g,n}\colon V^{\otimes n} \to H^\even(\oM_{g,n}),\quad 2g-2+n>0,
$$
such that the following axioms are satisfied:
\begin{enumerate}[(i)]
\item The maps $c_{g,n}$ are equivariant with respect to the $S_n$-action permuting the $n$ copies of~$V$ in $V^{\otimes n}$ and the $n$ marked points in $\oM_{g,n}$, respectively;

\item $\pi^* c_{g,n}( \otimes_{i=1}^n e_{\alpha_i}) = c_{g,n+1}(\otimes_{i=1}^n  e_{\alpha_i}\otimes e)$ and $c_{0,3}(e_{\alpha_1}\otimes e_{\alpha_2} \otimes e)=\eta_{\alpha_1\alpha_2}$;

\item $\gl_2^* c_{g_1+g_2,n_1+n_2}( \otimes_{i=1}^{n_1+n_2} e_{\alpha_i}) = c_{g_1,n_1+1}(\otimes_{i\in I_1} e_{\alpha_i} \otimes e_\mu)\otimes c_{g_2,n_2+1}(\otimes_{i\in I_2} e_{\alpha_i}\otimes e_\nu)\eta^{\mu \nu}$;

\item $\gl_1^* c_{g+1,n}(\otimes_{i=1}^n e_{\alpha_i}) = c_{g,n+2}(\otimes_{i=1}^n e_{\alpha_i}\otimes e_{\mu}\otimes e_\nu) \eta^{\mu \nu}$.
\end{enumerate}
In all axioms above we assume $\alpha_i\in\{1,\dots,N\}$ for any $i=1,2,\dots$. 
\end{definition}

For an arbitrary CohFT the formal power series
$$
F=F(t^1,\ldots,t^N)\coloneqq \sum_{n\geq 3}\frac{1}{n!}\sum_{1\leq\alpha_1,\ldots,\alpha_n\leq N}\left(\int_{\oM_{0,n}}c_{0,n}(\otimes_{i=1}^n e_{\alpha_i})\right)\prod_{i=1}^n t^{\alpha_i}
$$
satisfies the equations
\begin{align*}
A^\mu\frac{\d^3 F}{\d t^\mu\d t^\alpha \d t^\beta} &= \eta_{\alpha\beta}, && 1\leq \alpha,\beta\leq N,\\
\frac{\d^3 F}{\d t^\alpha \d t^\beta \d t^\mu} \eta^{\mu \nu}\frac{\d^3 F}{\d t^\nu \d t^\gamma \d t^\delta} &= \frac{\d^3 F}{\d t^\alpha \d t^\gamma \d t^\mu} \eta^{\mu \nu}\frac{\d^3 F}{\d t^\nu \d t^\beta \d t^\delta}, && 1\leq \alpha,\beta,\gamma,\delta\leq N.
\end{align*}
Thus, the formal power series $F$ defines a Dubrovin--Frobenius manifold structure on a formal neighbourhood of $0$ in $V$. 

\begin{definition}
A CohFT  $\{c_{g,n}\}$ is called \emph{semisimple} (at the origin) if the algebra defined by the structure constants $\left.c^\alpha_{\beta\gamma}\right|_{t^*=0}$, where 
$$
c^\alpha_{\beta\gamma}\coloneqq\eta^{\alpha\mu}\frac{\d^3F}{\d t^\mu\d t^\beta\d t^\gamma},
$$
doesn't have nilpotents.
\end{definition}

\subsubsection{Homogeneity}

Let $V$ be a graded vector space and assume the basis $e_1,\ldots,e_N$ is homogeneous with $\deg e_\alpha = q_\alpha$, $\alpha=1,\dots,N$. Assume also that $\deg e = 0$. 
By $\Deg\colon H^*(\oM_{g,n})\to H^*(\oM_{g,n})$ we denote the operator that acts on~$H^i(\oM_{g,n})$ by multiplication by $\frac{i}{2}$.

\begin{definition} A CohFT $\{c_{g,n}\}$ is called \emph{homogeneous}, or \emph{conformal}, if there exist complex constants $r^\alpha$, $\alpha=1,\dots,N$, and $\delta$ such that
\begin{gather}\label{eq:definition of a homogeneous CohFT}
\Deg c_{g,n}(\otimes_{i=1}^ne_{\alpha_i})+\pi_*c_{g,n+1}(\otimes_{i=1}^ne_{\alpha_i}\otimes r^\gamma e_\gamma)=\left(\sum_{i=1}^n q_{\alpha_i}+\delta(g-1)\right)c_{g,n}(\otimes_{i=1}^ne_{\alpha_i}).
\end{gather}
The constant $\delta$ is called the \emph{conformal dimension} of CohFT.
\end{definition}

For a homogeneous CohFT the formal power series $F(t^1,\ldots,t^N)$ satisfies the property
$$
\left((1-q_\alpha)t^\alpha+r^\alpha\right)\frac{\d F}{\d t^\alpha}=(3-\delta)F+\frac{1}{2}A_{\alpha\beta}t^\alpha t^\beta,
$$
where 
$$
A_{\alpha\beta}\coloneqq r^\mu c_{0,3}(e_\alpha\otimes e_\beta\otimes e_\mu).
$$
Thus, the associated Dubrovin--Frobenius manifold on a formal neighbourhood of~$0$ in $V$ is also homogeneous, with the Euler vector field given by
$$
E=E^\alpha\frac{\d}{\d t^\alpha}\coloneqq \left((1-q_\alpha)t^\alpha+r^\alpha\right)\frac{\d}{\d t^\alpha}.
$$


\subsection{Double ramification hierarchy}
\subsubsection{Formal loop space, differential polynomials, and local functionals}\label{subsubsection:formal loop space}
Introduce formal variables~$u^\alpha_i$, $\alpha=1,\ldots,N$, $i=0,1,\ldots$. Following~\cite{DZ01} (see also~\cite{Ros17}) we define the ring of \emph{differential polynomials} $\cA^0$ in the variables $u^1,\ldots,u^N$ as the ring of polynomials $f(u^*,u^*_x,u^*_{xx},\ldots)$ in the variables~$u^\alpha_i$, $i>0$, with coefficients in the ring of formal power series in the variables $u^\alpha=u^\alpha_0$:
$$
\cA^0\coloneqq\mbC[[u^*]][u^*_{\ge 1}].
$$

\begin{remark} This way we define a model of the loop space of the vector space $V$ by describing its ring of functions. In particular, it is useful to think of the variables $u^\alpha\coloneqq u^\alpha_0$ as the components $u^\alpha(x)$ of a formal loop $u\colon S^1\to V$ in the basis $e_1,\ldots,e_N$. Then the variables $u^\alpha_{1}\coloneqq u^\alpha_x, u^\alpha_{2}\coloneqq u^\alpha_{xx},\ldots$ are the components of the iterated $x$-derivatives of a formal loop.
\end{remark}

The \emph{standard gradation} on $\cA^0$, which we denote by $\deg$, is introduced by $\deg u^\alpha_i\coloneqq i$. The homogeneous component of $\cA^0$ of standard degree $d$ is denoted by $\cA^0_d$. The operator 
$$
\partial_x \coloneqq \sum_{i\geq 0} u^\alpha_{i+1}\frac{\partial}{\partial u^\alpha_i}
$$
increases the standard degree by $1$. Therefore, the quotient
$$
\Lambda^0\coloneqq\left.\cA^0\right/(\mbC\oplus\Im\,\d_x),
$$
called the space of \emph{local functionals}, inherits the standard gradation. The homogeneous component of $\Lambda^0$ of standard degree $d$ is denoted by $\Lambda^0_d$.  The natural projection $\cA^0\to\Lambda^0$ assigns to a differential polynomial $f$ the local functional $\overline{f}=\int f\,dx$.

The \emph{variational derivative} $\frac{\delta}{\delta u^\alpha}\colon \cA^0\to \cA^0$, $\alpha=1,\dots,N$, is defined by
$$
\frac{\delta}{\delta u^\alpha}\coloneqq\sum_{i\ge 0}(-\d_x)^i\circ\frac{\d}{\d u^\alpha_i}.
$$ 
Since it vanishes on $\mbC\oplus\Im\,\d_x$, it is well defined on the space of local functionals and abusing notation we denote it by the same symbol, $\frac{\delta}{\delta u^\alpha}\colon\Lambda^0\to\cA^0$.

We associate with a differential polynomial $f\in\cA^0$ a sequence of differential operators indexed by $\alpha=1,\ldots,N$ and $k\ge 0$: 
\begin{gather*}
L_\alpha^k(f)\coloneqq\sum_{i\ge k}{i\choose k}\frac{\d f}{\d u^\alpha_i}\d_x^{i-k}.
\end{gather*}
We often use the notation $L_\alpha(f)\coloneqq L^0_\alpha(f)$. 
These operators satisfy the property
\begin{gather*}
L_\alpha^k(\d_x f)=\d_x\circ L_\alpha^k(f)+L_\alpha^{k-1}(f),\quad k\ge 0,
\end{gather*}
where we adopt the convention $L_{\alpha}^{l}(f)\coloneqq 0$ for $l<0$. 

We associate with a local functional $\oh=\int h\,dx\in\Lambda^0$ a sequence of $N\times N$  matrices $\hOmega^k(\oh)=\left(\hOmega^k(\oh)^{\alpha\beta}\right)$ of differential operators, indexed by $k\ge 0$, defined as
$$
\hOmega^k(\oh)^{\alpha\beta}\coloneqq\eta^{\alpha\mu}\eta^{\beta\nu}L^k_\nu\left(\frac{\delta\oh}{\delta u^\mu}\right).
$$
We often use the notation $\hOmega(\oh)\coloneqq\hOmega^0(\oh)$.

Consider an $N\times N$ matrix~$K=(K^{\mu\nu})$ of differential operators of the form $K^{\mu\nu} = \sum_{j\geq 0} K^{\mu\nu}_j \partial_x^j$, where $K^{\mu\nu}_j\in\cA^0$ and the sum is finite. The conjugate matrix of differential operators $K^\dagger=\left((K^\dagger)^{\mu\nu}\right)$ is defined by
$$
(K^\dagger)^{\mu\nu}: = \sum_{j\geq 0} (-\partial_x)^j\circ K^{\nu\mu}_j.
$$

\begin{lemma}\label{lemma:property of hOmegak}
For any $\oh\in\Lambda^0$ we have $\hOmega^k(\oh)^\dagger=(-1)^k\hOmega^k(\oh)$.
\end{lemma}
\begin{proof} It is a straightforward computation.  We have to check that
$$
\left(L_\mu^k\left(\frac{\delta\oh}{\delta u^\nu}\right)\right)^\dagger=(-1)^kL_\nu^k\left(\frac{\delta\oh}{\delta u^\mu}\right).
$$
To this end, we compute
\begin{align*}
\left(L_\mu^k\left(\frac{\delta\oh}{\delta u^\nu}\right)\right)^\dagger
& = \sum_{i=0}^\infty \binom{k+i}{i} (-\d_x)^i \circ \frac{\d}{\d u^{\mu}_{k+i}} \frac{\delta\oh}{\delta u^\nu}  = \sum_{i,j=0}^\infty \binom{k+i+j}{i,j} (-\d_x)^i \frac{\d}{\d u^{\mu}_{k+i+j}} \frac{\delta\oh}{\delta u^\nu} (-\d_x)^j\\
& = \sum_{i,j,l=0}^\infty (-1)^{k+i+j} \binom{k+i+j+l}{i,j,l} (-\d_x)^{i+l} \frac{\d}{\d u^{\nu}_{k+i+j+l}} \frac{\delta\oh}{\delta u^\mu} (-\d_x)^j \\
& = (-1)^k \sum_{j=0}^\infty \binom{k+j}{j} \frac{\d}{\d u^{\nu}_{k+j}} \frac{\delta\oh}{\delta u^\mu} \d_x^j = (-1)^kL_\nu^k\left(\frac{\delta\oh}{\delta u^\mu}\right).
\end{align*}
Here, in order to pass to the second line in the computation we use the identity
\[
\frac{\d}{\d u^{\mu}_{t}} \frac{\delta}{\delta u^\nu} = (-1)^t \sum_{l=0}^\infty \binom{l+t}{t} (-\d_x)^l \frac{\d}{\d u^{\nu}_{t+l}}\frac{\delta}{\delta u^\mu}
\]
(cf.~\cite[Lemma 2.1.5(i)]{LZ11}), and in the third line we use that for any $t\geq 1$
\[
\sum_{i+l = t} \frac 1{i!l!} (-1)^i (-\d_x)^{i+l} = 0.
\]
\end{proof}

Differential polynomials and local functionals can also be described using another set of formal variables, corresponding heuristically to the Fourier components $p^\alpha_k$, $k\in\mbZ$, of the functions $u^\alpha=u^\alpha(x)$.  We define a change of variables
\begin{gather}\label{eq:u-p change}
u^\alpha_j = \sum_{k\in\mbZ} (i k)^j p^\alpha_k e^{i k x}, 
\end{gather}
which allows us to express a differential polynomial $f(u,u_x,u_{xx},\ldots)\in \cA^0$ as a formal Fourier series in $x$. In the latter expression the coefficient of $e^{i k x}$ is a power series in the variables~$p^\alpha_j$ with the sum of the subscripts in each monomial in $p^\alpha_j$ equal to $k$. Moreover, the local functional~$\overline{f}$ corresponds to the constant term of the Fourier series of $f$.

\subsubsection{Poisson brackets and Hamiltonian hierarchies}

Let us describe a natural class of Poisson brackets on the space of local functionals. Given an $N\times N$ matrix~$K=(K^{\mu\nu})$ of differential operators of the form $K^{\mu\nu} = \sum_{j\geq 0} K^{\mu\nu}_j \partial_x^j$, where $K^{\mu\nu}_j\in\cA^0$ and the sum is finite, we define
$$
\{\overline{f},\overline{g}\}_{K}\coloneqq\int\left(\frac{\delta \overline{f}}{\delta u^\mu}K^{\mu \nu}\frac{\delta \overline{g}}{\delta u^\nu}\right)dx.
$$
The bracket $\{\cdot,\cdot\}_K$ is skew-symmetric if and only if $K^\dagger=-K$. The bracket $\{\cdot,\cdot\}_K$ satisfies the Jacobi identity if and only if the Schouten--Nijenhuis bracket of the bracket with itself vanishes,
$$
\left[\{\cdot,\cdot\}_K,\{\cdot,\cdot\}_K\right]=0,
$$
which is discussed in details in Section~\ref{subsection:schouten bracket}. 

\begin{definition}
An operator $K$ is called \emph{Poisson}, if the bracket $\{\cdot,\cdot\}_K$ is skew-symmetric and satisfies the Jacobi identity.	
\end{definition}

\begin{example}
 A standard example of a Poisson operator is given by the operator $\eta^{-1}\partial_x$. The corresponding Poisson bracket has a nice expression in terms of the variables $p^\alpha_k$:
\begin{gather*}
\{p^\alpha_k, p^\beta_j\}_{\eta^{-1}\partial_x} = i k \eta^{\alpha \beta} \delta_{k+j,0}.
\end{gather*}
\end{example}

Consider the extensions $\hcA^0\coloneqq \cA^0\otimes \mbC[[\eps]]$ and $\hLambda^0\coloneqq \Lambda^0\otimes \mbC[[\eps]]$ of the spaces $\cA^0$ and $\Lambda^0$ with a new variable~$\eps$ of standard gradation $\deg\eps\coloneqq -1$. Let $\hcA^0_k$ and~$\hLambda^0_k$ denote the subspaces of degree~$k$ of $\hcA^0$ and $\hLambda^0$, respectively. Abusing the terminology we still call their elements \emph{differential polynomials} and \emph{local functionals}. 

We can also define a bracket $\{\cdot,\cdot\}_K\colon\hLambda^0\times\hLambda^0\to\hLambda^0$ as above starting from an operator $K=(K^{\mu\nu})$, $K^{\mu\nu} = \sum_{i,j\ge 0} K^{[i],\mu\nu}_j\eps^i\partial_x^j$, $K^{[i],\mu\nu}_j\in\cA^0$, where for each $i\ge 0$ the sum $\sum_{j\ge 0} K^{[i],\mu\nu}_j\d_x^j$ is finite. An operator $K=(K^{\alpha\beta})$ and  the corresponding bracket $\{\cdot,\cdot\}_K$ have degree $d$ if $K^{\alpha\beta}=\sum_{s\ge 0}K^{\alpha\beta}_s\d_x^s$, where $K^{\alpha\beta}_s\in\hcA^0_{-s+d}$. 

\begin{definition}A {\it Hamiltonian hierarchy} of PDEs is a system of the form
\begin{gather}\label{eq:Hamiltonian system}
\frac{\partial u^\alpha}{\partial \tau_i} = K^{\alpha\mu} \frac{\delta\overline{h}_i}{\delta u^\mu}, \quad 1\le\alpha\le N,\quad i\ge 1,
\end{gather}
where $\oh_i\in\hLambda^0_0$, $K=(K^{\mu\nu})$ is a Poisson operator of degree $1$, and the compatibility condition $\{\oh_i,\oh_j\}_K=0$ for $i,j\geq 1$ is satisfied. 
	
The local functionals~$\oh_i$ are called the {\it Hamiltonians} of the system~\eqref{eq:Hamiltonian system}. 
\end{definition}

Consider Hamiltonian hierarchies of the form
\begin{gather}\label{eq:Hamiltonian system,2}
\frac{\d u^\alpha}{\d t^\beta_q} = K_1^{\alpha\mu}\frac{\delta\oh_{\beta,q}}{\delta u^\mu}, \quad 1\le\alpha,\beta\le N,\quad q\ge 0,
\end{gather}
equipped with $N$ linearly independent Casimirs $\oh_{\alpha,-1}$, $1\le\alpha\le N$, of the Poisson bracket~$\{\cdot,\cdot\}_{K_1}$. Two Poisson operators $K_1$ and $K_2$ are said to be {\it compatible} if the linear combination $K_2-\lambda K_1$ is a Poisson operator for any $\lambda\in\mbC$. 

\begin{definition} A Hamiltonian hierarchy~\eqref{eq:Hamiltonian system,2} is said to be {\it bihamiltonian} if it is endowed with a Poisson operator~$K_2$ of degree~$1$ compatible with the operator~$K_1$ and such that
\begin{gather}\label{eq:bihamiltonian recursion}
\{\cdot,\oh_{\alpha,i-1}\}_{K_2}=\sum_{j=0}^i R^{j,\beta}_{i,\alpha}\{\cdot,\oh_{\beta,i-j}\}_{K_1},\quad 1\le\alpha\le N,\quad i\ge 0,
\end{gather}
where $R^j_i=(R^{j,\beta}_{i,\alpha})$, $0\le j\le i$, are constant $N\times N$ matrices.
	
The relation~\eqref{eq:bihamiltonian recursion} is called a {\it bihamiltonian recursion}. 
\end{definition}

\begin{remark}
Note that if the matrices $R^0_i$, $i\ge 0$, are invertible, then the bihamiltonian recursion~\eqref{eq:bihamiltonian recursion} determines the local functionals $\oh_{\alpha,i}$, $i\ge 0$, uniquely up to a triangular transformation
$$
\oh_{\alpha,i}\mapsto\oh_{\alpha,i}+\sum_{j=1}^{i+1} A^{j,\beta}_\alpha\oh_{\beta,i-j},\quad 1\le\alpha\le N,\quad i\ge 0,\qquad A^{j,\beta}_\alpha\in\mbC. 
$$
\end{remark}

In the bihamiltonian hierarchies considered below all but finitely many matrices $R^0_i$ are invertible.

\subsubsection{Construction of the double ramification hierarchy} Denote by $\psi_i\in H^2(\oM_{g,n})$ the first Chern class of the line bundle over~$\oM_{g,n}$ formed by the cotangent lines at the $i$-th marked point. Denote by~$\mathbb E$ the rank~$g$ Hodge vector bundle over~$\oM_{g,n}$ whose fibers are the spaces of holomorphic one-forms on stable curves. Let $\lambda_j\coloneqq c_j(\mathbb E)\in H^{2j}(\oM_{g,n})$. 

For any $a_1,\dots,a_n\in \mbZ$, $\sum_{i=1}^n a_i =0$, denote by $\DR_g(a_1,\ldots,a_n) \in H^{2g}(\oM_{g,n})$ the {\it double ramification (DR) cycle}. We refer the reader, for example, to~\cite{BSSZ15} for the definition of the DR cycle on~$\oM_{g,n}$, which is based on the notion of a stable map to $\CP^1$ relative to~$0$ and~$\infty$. If not all the multiplicities $a_i$ are equal to zero, then one can think of the class $\DR_g(a_1,\ldots,a_n)$ as the Poincar\'e dual to a compactification in~$\oM_{g,n}$ of the locus of pointed smooth curves~$(C;p_1,\ldots,p_n)$ satisfying $\mathcal O_C\left(\sum_{i=1}^n a_ip_i\right)\cong\mathcal O_C$. Consider the Poincar\'e dual to the double ramification cycle~$\DR_g(a_1,\ldots,a_n)$ in the space $\oM_{g,n}$. It is an element of $H_{2(2g-3+n)}(\oM_{g,n})$, and abusing notation it is also denoted by $\DR_g(a_1,\ldots,a_n)$. In genus $0$ we have $\DR_0(a_1,\ldots,a_n)=1\in H^0(\oM_{0,n})$. 

Consider a CohFT $\{c_{g,n}\colon V^{\otimes n} \to H^{\even}(\oM_{g,n})\}$. The Hamiltonians of the \emph{double ramification hierarchy} are defined as follows:
\begin{gather}\label{DR Hamiltonians}
\og_{\alpha,d}\coloneqq\sum_{\substack{g\ge 0\\n\ge 2}}\frac{(-\eps^2)^g}{n!}\sum_{\substack{a_1,\ldots,a_n\in\mbZ\\\sum a_i=0}}\left(\int_{\DR_g(0,a_1,\ldots,a_n)}\lambda_g\psi_1^d c_{g,n+1}(e_\alpha\otimes\otimes_{i=1}^n e_{\alpha_i})\right)\prod_{i=1}^n p^{\alpha_i}_{a_i},
\end{gather}
for $\alpha=1,\ldots,N$ and $d=0,1,2,\ldots$. 

The expression on the right-hand side of~\eqref{DR Hamiltonians} can be uniquely written as a local functional from $\hLambda^0_0$ using the change of variables~\eqref{eq:u-p change}. Concretely, it can be done in the following way. The restriction $\DR_g(a_1,\ldots,a_n)\big|_{\cM_{g,n}^{\ct}}$, where $\cM_{g,n}^{\ct}$ is the moduli space of stable curves of compact type, is a homogeneous polynomial in $a_1,\ldots,a_n$ of degree $2g$ with the coefficients in $H^{2g}(\cM_{g,n}^{\ct})$. This follows from Hain's formula~\cite{Hai13} for the version of the DR cycle defined using the universal Jacobian over~$\cM^\ct_{g,n}$ and the result of the paper~\cite{MW13}, where it is proved that the two versions of the DR cycle coincide on $\cM^\ct_{g,n}$ (the polynomiality of the DR cycle on~$\oM_{g,n}$ is proved in~\cite{JPPZ17}). The polynomiality of the DR cycle on $\cM^\ct_{g,n}$ together with the fact that~$\lambda_g$ vanishes on~$\oM_{g,n}\setminus\cM_{g,n}^{\ct}$ (see e.g.~\cite[Section~0.4]{FP00}) imply that the integral
\begin{gather}\label{DR integral}
\int_{\DR_g(0,a_1,\ldots,a_n)}\lambda_g\psi_1^d c_{g,n+1}(e_\alpha\otimes\otimes_{i=1}^n e_{\alpha_i})
\end{gather}
is a homogeneous polynomial in $a_1,\ldots,a_n$ of degree~$2g$,  which we denote by 
\begin{equation*}
P_{\alpha,d,g;\alpha_1,\ldots,\alpha_n}(a_1,\ldots,a_n)=\sum_{\substack{b_1,\ldots,b_n\ge 0\\b_1+\ldots+b_n=2g}}P_{\alpha,d,g;\alpha_1,\ldots,\alpha_n}^{b_1,\ldots,b_n}a_1^{b_1}\ldots a_n^{b_n}.
\end{equation*}
Then we have
$$
\og_{\alpha,d}=\int\sum_{\substack{g\ge 0\\n\ge 2}}\frac{\eps^{2g}}{n!}\sum_{\substack{b_1,\ldots,b_n\ge 0\\b_1+\ldots+b_n=2g}}P_{\alpha,d,g;\alpha_1,\ldots,\alpha_n}^{b_1,\ldots,b_n} u^{\alpha_1}_{b_1}\ldots u^{\alpha_n}_{b_n}dx.
$$

Note that the integral~\eqref{DR integral} is defined only when $a_1+\cdots+a_n=0$. Therefore, the polynomial~$P_{\alpha,d,g;\alpha_1,\ldots,\alpha_n}$ is not uniquely defined. However, the resulting local functional $\og_{\alpha,d}\in\hLambda^0_0$ doesn't depend on this ambiguity, see~\cite{Bur15}. In fact, in \cite{BR16} a special choice of the differential polynomial densities $g_{\alpha,d} \in \hcA^0_0$ for $\og_{\alpha,d} = \int g_{\alpha,d}\,dx$ is proposed. The densities are defined in terms of the $p$-variables as
$$
g_{\alpha,d}\coloneqq\sum_{\substack{g\ge 0,\,n\ge 1\\2g-1+n>0}}\frac{(-\eps^2)^g}{n!}\sum_{\substack{a_0,\ldots,a_n\in\mbZ\\\sum a_i=0}}\left(\int_{\DR_g(a_0,a_1,\ldots,a_n)}\lambda_g\psi_1^d c_{g,n+1}(e_\alpha\otimes \otimes_{i=1}^n e_{\alpha_i})\right)\prod_{i=1}^n p^{\alpha_i}_{a_i} e^{-i a_0 x}
$$
and converted unequivocally to differential polynomials using the change of variables~(\ref{eq:u-p change}).

It is proved in~\cite{Bur15} that the local functionals~$\og_{\alpha,d}$ mutually commute with respect to the standard bracket~$\{\cdot,\cdot\}_{\eta^{-1}\d_x}$.

\begin{definition}
 The system of local functionals $\og_{\alpha,d}$, for $\alpha=1,\ldots,N$, $d=0,1,2,\ldots$, and the corresponding system of Hamiltonian PDEs with respect to the Poisson bracket~$\{\cdot,\cdot\}_{\eta^{-1}\partial_x}$,
$$
\frac{\d u^\alpha}{\d t^\beta_q}=\eta^{\alpha\mu}\d_x\frac{\delta\og_{\beta,q}}{\delta u^\mu},\quad 1\le\alpha,\beta\le N,\quad q\ge 0,
$$
is called the \emph{double ramification hierarchy} (or the \emph{DR hierarchy} for brevity). 	
\end{definition}

\subsubsection{Extra structures for the double ramification hierarchy}
We can equip the DR hierarchy with the following~$N$ linearly independent Casimirs of its Poisson bracket $\{\cdot,\cdot\}_{\eta^{-1}\d_x}$:
$$
\og_{\alpha,-1}\coloneqq\int\eta_{\alpha\beta}u^\beta dx,\quad 1\le\alpha\le N.
$$

Another important object related to the DR hierarchy is the local functional $\og$ defined in terms of the $p$-variables as
\begin{gather*}
\og\coloneqq\sum_{\substack{g,n\ge 0\\2g-2+n>0}}\frac{(-\eps^2)^g}{n!}\sum_{\substack{a_1,\ldots,a_n\in\mbZ\\\sum a_i=0}}\left(\int_{\DR_g(a_1,\ldots,a_n)}\lambda_gc_{g,n}(\otimes_{i=1}^n e_{\alpha_i})\right)\prod_{i=1}^n p^{\alpha_i}_{a_i}.
\end{gather*}
We have the following relations~\cite[Section 4.2.5]{Bur15}:
\begin{gather}\label{eq:properties of og}
\og_{\un,1}=(D-2)\og\qquad \text{and} \qquad\frac{\d\og}{\d u^\alpha}=\og_{\alpha,0},
\end{gather}
where $D\coloneqq\sum_{n\ge 0}(n+1)u^\alpha_n\frac{\d}{\d u^\alpha_n}$ and $\og_{\un,1}\coloneqq A^\alpha\og_{\alpha,1}$. Also, the local functional~$\og$ has the following explicit expression at the approximation up to~$\eps^2$ (\cite[Lemma~8.1]{BDGR18}):
\begin{gather}\label{eq:og up to genus 1}
\og=\int \left.F\right|_{t^*=u^*}dx-\frac{\eps^2}{48}\int\left.\left(c^\theta_{\theta\xi}c^\xi_{\alpha\beta}\right)\right|_{t^*=u^*}u^\alpha_xu^\beta_x dx+O(\eps^4).
\end{gather}


\subsection{Bihamiltonian structure for the double ramification hierarchy}

Consider a homogeneous cohomological field theory and the associated DR hierarchy. Let 
$$
\mu_\alpha\coloneqq q_\alpha-\frac{\delta}{2}
$$
and define a diagonal matrix $\mu$ by
$$
\mu\coloneqq\diag(\mu_1,\ldots,\mu_N).
$$
Note that 
\begin{gather}\label{eq:mu and eta}
\mu\eta+\eta\mu=0.
\end{gather}
Introduce an operator $\hE$ by
$$
\hE\coloneqq\sum_{n\ge 0}\left((1-q_\alpha)u^\alpha_n+\delta_{n,0}r^\alpha\right)\frac{\d}{\d u^\alpha_n}+\frac{1-\delta}{2}\eps\frac{\d}{\d\eps}.
$$

\begin{definition}
Define an operator $K_2$ of degree~$1$ by
\begin{gather*}
K_2\coloneqq\hE\left(\hOmega(\og)\right)\circ\d_x+\hOmega(\og)_x\circ\left(\frac{1}{2}-\mu\right)+\d_x\circ\hOmega^1(\og)\circ\d_x,
\end{gather*}
where the notation $\hE\big(\hOmega(\og)\big)$ (respectively, $\hOmega(\og)_x$) means that we apply the operator $\hE$ (respectively, $\d_x$) to the coefficients of the operator $\hOmega(\og)$.
\end{definition}

\begin{conjecture}\label{main conjecture} {\ }
\begin{enumerate}
	\item The operator $K_2$ is Poisson and compatible with the operator $K_1\coloneqq\eta^{-1}\d_x$.
	\item The Poisson brackets $\{\cdot,\cdot\}_{K_2}$ and $\{\cdot,\cdot\}_{K_1}$ give a bihamiltonian structure for the DR hierarchy with the following bihamiltonian recursion:
	\begin{gather}\label{eq:DR biham}
	\left\{\cdot,\og_{\alpha,d}\right\}_{K_2}=\left(d+\frac{3}{2}+\mu_\alpha\right)\left\{\cdot,\og_{\alpha,d+1}\right\}_{K_1}+A^\beta_\alpha\left\{\cdot,\og_{\beta,d}\right\}_{K_1},\quad d\ge -1,
	\end{gather}
	where $A^\beta_\alpha\coloneqq\eta^{\beta\nu}A_{\nu\alpha}$.
\end{enumerate}
\end{conjecture}

\begin{lemma} The operator $K_2$ has the following alternative expression:
\begin{gather}\label{eq:second expression for K2}
K_2=\d_x\circ\hOmega(\og)\circ\left(\frac{1}{2}-\mu\right)+\left(\frac{1}{2}-\mu\right)\circ\hOmega(\og)\circ\d_x+\eta^{-1}A\eta^{-1}\d_x+\d_x\circ\hOmega^1(\og)\circ\d_x,
\end{gather}
where $A=(A_{\alpha\beta})$.
\end{lemma}
\begin{proof}
Equation~\eqref{eq:second expression for K2} is equivalent to
\begin{align*}
\hE\left(\hOmega(\og)\right)=&\hOmega(\og)\circ\left(\frac{1}{2}-\mu\right)+\left(\frac{1}{2}-\mu\right)\circ\hOmega(\og)+\eta^{-1}A\eta^{-1}\,,
\end{align*}
which is equivalent to
\begin{align*}
\hE\left(L_\alpha\left(\frac{\delta\og}{\delta u^\beta}\right)\right)=&(q_\alpha+q_\beta+1-\delta)L_\alpha\left(\frac{\delta\og}{\delta u^\beta}\right)+A_{\alpha\beta}.
\end{align*}

Multiplying both sides of Equation~\eqref{eq:definition of a homogeneous CohFT} by $\lambda_g\DR_g(a_1,\ldots,a_n)$, integrating over $\oM_{g,n}$, and taking the generating series we obtain
$$
\left(\frac{\eps}{2}\frac{\d\og}{\d\eps}-3\og+\sum_{a\in\mbZ}p^\alpha_a\frac{\d\og}{\d p^\alpha_a}\right)+\left(r^\alpha\frac{\d\og}{\d u^\alpha}-\int\frac{1}{2}A_{\alpha\beta}u^\alpha u^\beta dx\right)=\sum_{a\in\mbZ}q_\alpha p^\alpha_a\frac{\d\og}{\d p^\alpha_a}+\left(\frac{\delta}{2}\eps\frac{\d\og}{\d\eps}-\delta\og\right),
$$
where we used that, since the (complex) dimension of $\oM_{g,n}$ is equal to $3g-3+n$, we have $\int_{\oM_{g,n}}\lambda_g\DR_g(a_1,\ldots,a_n)\Deg c_{g,n}(\otimes_{i=1}^n e_{\alpha_i})=(g-3+n)\int_{\oM_{g,n}}\lambda_g\DR_g(a_1,\ldots,a_n)c_{g,n}(\otimes_{i=1}^n e_{\alpha_i})$. Therefore, we get
$$
\sum_{a\in\mbZ}(1-q_\alpha)p^\alpha_a\frac{\d\og}{\d p^\alpha_a}+r^\alpha\frac{\d\og}{\d u^\alpha}+\frac{1-\delta}{2}\eps\frac{\d\og}{\d\eps}=(3-\delta)\og.
$$
Noting that in the $u$-variables the operator $\sum_{a\in\mbZ}(1-q_\alpha)p^\alpha_a\frac{\d}{\d p^\alpha_a}$ is equal to $\sum_{k\ge 0}(1-q_\alpha)u^\alpha_k\frac{\d}{\d u^\alpha_k}$, we obtain
$$
\hE\og=(3-\delta)\og+\int\frac{1}{2}A_{\alpha\beta}u^\alpha u^\beta dx.
$$

This implies
\begin{align*}
\hE\left(L_\alpha\left(\frac{\delta\og}{\delta u^\beta}\right)\right)
& =L_\alpha\left(\frac{\delta(\hE\og)}{\delta u^\beta}\right)-(2-q_\alpha-q_\beta)L_\alpha\left(\frac{\delta\og}{\delta u^\beta}\right)
\\
& =(q_\alpha+q_\beta+1-\delta)L_\alpha\left(\frac{\delta\og}{\delta u^\beta}\right)+A_{\alpha\beta}\,,
\end{align*}
as required.
\end{proof}

\begin{remark} Formula~\eqref{eq:second expression for K2} together with Lemma~\ref{lemma:property of hOmegak} immediately implies that the operator~$K_2$ is skew-symmetric.
\end{remark}


\section{Evidence I: checks for a general CohFT}\label{section:general checks}

In this section we check some parts of Conjecture~\ref{main conjecture} for an arbitrary homogeneous CohFT~$\{c_{g,n}\}$: we prove it in genus $0$, examine the first step of the bihamiltonian recursion, and compute the Schouten--Nijenhuis bracket of $\{\cdot,\cdot\}_{K_2}$ and $\{\cdot,\cdot\}_{K_1}$ (considered as bivector fields).

\subsection{Genus $0$} By reduction to genus $0$, we mean the reduction to the dispersionless part of the hierarchy, that is, we set $\eps=0$.

\begin{proposition}
Conjecture~\ref{main conjecture} is true in genus~$0$.
\end{proposition}
\begin{proof}
Recall that $F=F(t^1,\ldots,t^N)$ denotes the Dubrovin--Frobenius manifold potential associated to a {CohFT} $\{c_{g,n}\}$. We have
$$
\og|_{\eps=0}=\int F(u^1,\ldots,u^N)dx.
$$
Therefore,
$$
K_2^{[0],\alpha\beta}\coloneqq \left.K_2^{\alpha\beta}\right|_{\eps=0}=g^{\alpha\beta}\d_x+\d_x\Omega^{\alpha\beta}\left(\frac{1}{2}-\mu_\beta\right),
$$
where
\begin{gather}
g^{\alpha\beta}\coloneqq\left.\eta^{\alpha\mu}\eta^{\beta\nu}E^\gamma\frac{\d^3 F}{\d t^\gamma\d t^\mu\d t^\nu}\right|_{t^\theta=u^\theta}
\qquad\text{and}\qquad \Omega^{\alpha\beta}\coloneqq\left.\eta^{\alpha\mu}\eta^{\beta\nu}\frac{\d^2 F}{\d t^\mu\d t^\nu}\right|_{t^\theta=u^\theta}.\label{eq:galphabeta}
\end{gather}
The fact that the operators $\eta^{-1}\d_x$ and $K_2^{[0]}$ form a pair of compatible Poisson operators is well known (see e.g.~\cite[page~443]{DZ99}). This proves the first part of the conjecture in genus $0$. 

Now we check the second part of the conjecture. Let $t^\alpha_a$, $1\le\alpha\le N$, $a\ge 0$, be formal variables, where we identify $t^\alpha_0=t^\alpha$. Consider the genus $0$ potential of $\{c_{g,n}\}$:
$$
\mcF_0(t^*_*)\coloneqq\sum_{n\geq 3}\frac{1}{n!}\sum_{\substack{1\leq\alpha_1,\ldots,\alpha_n\leq N\\d_1,\ldots,d_n\ge 0}}\left(\int_{\oM_{0,n}}c_{0,n}(\otimes_{i=1}^n e_{\alpha_i})\prod_{i=1}^n\psi_i^{d_i}\right)\prod_{i=1}^n t^{\alpha_i}_{d_i}\in\mbC[[t^*_*]].
$$
It satisfies the string equation
\begin{gather}\label{eq:string equation for mcF0}
\frac{\d\mcF_0}{\d t^\un_0}=\sum_{n\ge 0}t^\alpha_{n+1}\frac{\d\mcF_0}{\d t^\alpha_n}+\frac{1}{2}\eta_{\alpha\beta}t^\alpha_0 t^\beta_0,
\end{gather}
where we use the notation $\frac{\d}{\d t^\un_0}\coloneqq A^\alpha\frac{\d}{\d t^\alpha_0}$, and the topological recursion relations
\begin{gather}\label{eq:TRR-0}
\frac{\d^3\mcF_0}{\d t^\alpha_{a+1}\d t^\beta_b\d t^\gamma_c}=\frac{\d^2\mcF_0}{\d t^\alpha_a\d t^\mu_0}\eta^{\mu\nu}\frac{\d^3\mcF_0}{\d t^\nu_0\d t^\beta_b\d t^\gamma_c},\quad 1\le\alpha,\beta,\gamma\le N,\quad a,b,c\ge 0.
\end{gather}
Also, the homogeneity condition~\eqref{eq:definition of a homogeneous CohFT} implies that
\begin{gather}\label{eq:homogeneity of mcF0}
\left(\sum_{n\ge 0}(1-q_\alpha-n)t^\alpha_n\frac{\d}{\d t^\alpha_n}+r^\alpha\frac{\d}{\d t^\alpha_0}-\sum_{n\ge 0}A^\alpha_\beta t^\beta_{n+1}\frac{\d}{\d t^\alpha_n}\right)\mcF_0=(3-\delta)\mcF_0+\frac{1}{2}A_{\alpha\beta}t^\alpha_0 t^\beta_0.
\end{gather}
Denote
$$
\Omega_{\alpha,a;\beta,b}\coloneqq\left.\frac{\d^2\mcF_0}{\d t^\alpha_a\d t^\beta_b}\right|_{t^\gamma_c=\delta_{c,0}u^\gamma}.
$$
We obtain
$$
\oh_{\alpha,d}\coloneqq\og_{\alpha,d}|_{\eps=0}=\int\left(\left.\frac{\d\mcF_0}{\d t^\alpha_d}\right|_{t^\gamma_c=\delta_{c,0}u^\gamma}\right)dx\xlongequal{\!\text{Eq.~\eqref{eq:string equation for mcF0}}\!}\int\Omega_{\un,0;\alpha,d+1}dx.
$$
Note that by Equation~\eqref{eq:string equation for mcF0} we have
\begin{equation}\label{eq:Omega-String}
\frac{\d\Omega_{\un,0;\alpha,d+1}}{\d u^\gamma}=\Omega_{\gamma,0;\alpha,d},\quad d\ge -1,
\end{equation}
where we adopt the convention~$\Omega_{\gamma,0;\alpha,-1}\coloneqq\eta_{\gamma\alpha}$. We need to check that for $d\ge -1$ we have
\begin{gather}\label{eq:bihamiltonian recursion in genus 0}
\{\cdot,\oh_{\alpha,d}\}_{K_2^{[0]}}=\left(d+\frac{3}{2}+\mu_\alpha\right)\{\cdot,\oh_{\alpha,d+1}\}_{\eta^{-1}\d_x}+A^\beta_\alpha\{\cdot,\oh_{\beta,d}\}_{\eta^{-1}\d_x}.
\end{gather}

Equation~\eqref{eq:bihamiltonian recursion in genus 0} is equivalent to
\begin{align} \label{eq:EquivGenus0}
&K_2^{[0],\beta\gamma}\frac{\d\Omega_{\un,0;\alpha,d+1}}{\d u^\gamma}=\left(d+\frac{3}{2}+\mu_\alpha\right)\eta^{\beta\gamma}\d_x\frac{\d\Omega_{\un,0;\alpha,d+2}}{\d u^\gamma}+\eta^{\beta\gamma}\d_x\frac{\d\Omega_{\un,0;\mu,d+1}}{\d u^\gamma}A^\mu_\alpha.
\end{align}
Using Equation~\eqref{eq:Omega-String} and the skew-symmetry of the operator $K_2^{[0]}$ we have
\[
K_2^{[0],\beta\gamma}\frac{\d\Omega_{\un,0;\alpha,d+1}}{\d u^\gamma} = K_2^{[0],\beta\gamma}\Omega_{\gamma,0;\alpha,d}
= \left(\d_x\circ g^{\beta\gamma}+\left(\mu_\beta-\frac{1}{2}\right)\d_x\Omega^{\beta\gamma}\right)\Omega_{\gamma,0;\alpha,d}.
\]
By~\eqref{eq:TRR-0} we have $g^{\beta\gamma}\Omega_{\gamma,0;\alpha,d}=\eta^{\beta\gamma}((1-q_\nu)u^\nu+r^\nu)\frac{\d\Omega_{\gamma,0;\alpha,d+1}}{\d u^\nu}$ and $\d_x\Omega^{\beta\gamma}\Omega_{\gamma,0;\alpha,d}=\eta^{\beta\gamma}\d_x\Omega_{\gamma,0;\alpha,d+1}$. Recall also that $\mu\eta+\eta\mu=0$. Therefore, the left-hand side of Equation~\eqref{eq:EquivGenus0} is equal to
\[
\eta^{\beta\gamma}\d_x \bigg( ((1-q_\nu)u^\nu+r^\nu)\frac{\d\Omega_{\gamma,0;\alpha,d+1}}{\d u^\nu}  +\left(-\mu_\gamma-\frac{1}{2}\right)\d_x\Omega_{\gamma,0;\alpha,d+1}\bigg).
\]
Using Equation~\eqref{eq:Omega-String} we rewrite the right-hand side of Equation~\eqref{eq:EquivGenus0} as
\[
\eta^{\beta\gamma}\d_x \bigg(\left(d+\frac{3}{2}+\mu_\alpha\right)\Omega_{\gamma,0;\alpha,d+1}+\Omega_{\gamma,0;\mu,d}A^\mu_\alpha\bigg).
\]
Therefore, Equation~\eqref{eq:EquivGenus0} would follow from 
\begin{gather}\label{eq:proof of genus 0 conjecture,1}
((1-q_\nu)u^\nu+r^\nu)\frac{\d\Omega_{\gamma,0;\alpha,d+1}}{\d u^\nu}=\left(d+2+\mu_\alpha+\mu_\gamma\right)\Omega_{\gamma,0;\alpha,d+1}+\Omega_{\gamma,0;\mu,d}A^\mu_\alpha.
\end{gather}
Applying $\frac{\d^2}{\d t^\gamma_0\d t^\alpha_{d+1}}$ to both sides of Equation~\eqref{eq:homogeneity of mcF0} and setting $t^\nu_p=\delta_{p,0}u^\nu$ we prove Equation~\eqref{eq:proof of genus 0 conjecture,1}, hence Equation~\eqref{eq:EquivGenus0}, hence Equation~\eqref{eq:bihamiltonian recursion in genus 0}. This completes the proof of the proposition.
\end{proof}

\begin{remark}
The genus $0$ part of the DR hierarchy coincides with a principal hierarchy (a genus~$0$ hierarchy in the language of~\cite{DZ98}) associated to the Dubrovin--Frobenius manifold given by the potential $F$~\cite[Section~4.2.2]{Bur15}. The family of principal hierarchies associated to a Dubrovin--Frobenius manifold is parameterized by the choice of deformed flat coordinates $\theta^\alpha(t^*,z)$ for a certain flat connection on a formal neighbourhood of $0\in\mbC^N$ (see e.g.~\cite[Equation~(2.38)]{DZ98}). In the literature the deformed flat coordinates are often chosen in such a way that they satisfy an additional quasihomogeneity constraint (see e.g.~\cite[Equation~(2.74)]{DZ98}). An analogue of relation~\eqref{eq:bihamiltonian recursion in genus 0} for the corresponding principal hierarchy is well known~\cite[Proposition~2 and Equation~(2.75)]{DZ98}. The genus $0$ part of the DR hierarchy corresponds to a unique choice of the deformed flat coordinates $\theta^\alpha(t^*,z)=t^\alpha+\sum_{p\ge 1}\theta^\alpha_p(t^*)z^p$ satisfying $\frac{\d\theta^\alpha_p}{\d t^\beta}\big|_{t^*=0}=0$. 
\end{remark}


\subsection{The first equation of the bihamiltonian recursion}

\begin{proposition}
Equation~\eqref{eq:DR biham} is true for $d=-1$.
\end{proposition}
\begin{proof}
For $d=-1$ Equation~\eqref{eq:DR biham} is equivalent to
\begin{gather*}
K_2^{\beta\gamma}\frac{\delta}{\delta u^\gamma}\left(\int\eta_{\alpha\nu}u^\nu dx\right)=\left(\mu_\alpha+\frac{1}{2}\right)\eta^{\beta\gamma}\d_x\frac{\delta\og_{\alpha,0}}{\delta u^\gamma}+A_\alpha^\nu\eta^{\beta\gamma}\d_x\frac{\delta}{\delta u^\gamma}\left(\int\eta_{\nu\theta}u^\theta dx\right).
\end{gather*}
The variational derivative of $\int u^* dx$ is constant, therefore, expanding the definition of $K_2$ we see that the latter equation is equivalent to
\begin{gather*}
\hOmega(\og)^{\beta\gamma}_x\left(\left(\frac{1}{2}-\mu_\gamma\right)\eta_{\gamma\alpha}\right)=\left(\mu_\alpha+\frac{1}{2}\right)\eta^{\beta\gamma}\d_x\frac{\delta\og_{\alpha,0}}{\delta u^\gamma}.
\end{gather*}
Expanding the definition of $\hOmega$ and using the skew-symmetry of $\mu$ with respect to $\eta$ we see that the latter equation is equivalent to
\begin{align*}
\left(\mu_\alpha+\frac{1}{2}\right)\eta^{\beta\gamma}\d_x\frac{\d}{\d u^\alpha}\left(\frac{\delta\og}{\delta u^\gamma}\right)=\left(\mu_\alpha+\frac{1}{2}\right)\eta^{\beta\gamma}\d_x\frac{\delta\og_{\alpha,0}}{\delta u^\gamma}.
\end{align*}
This last equation holds, since 
\[
\frac{\d}{\d u^\alpha}\left(\frac{\delta\og}{\delta u^\gamma}\right)=\frac{\delta}{\delta u^\gamma}\left(\frac{\d\og}{\d u^\alpha}\right)\xlongequal{\!\text{Eq.~\eqref{eq:properties of og}}\!}\frac{\delta\og_{\alpha,0}}{\delta u^\gamma}.
\] 
\end{proof}


\subsection{The Schouten--Nijenhuis bracket of two skew-symmetric operators}\label{subsection:schouten bracket}

In this section we prove that the Schouten--Nijenhuis bracket of the brackets $\{\cdot,\cdot\}_{K_1}$ and $\{\cdot,\cdot\}_{K_2}$ considered as bivector fields vanishes. In order to formulate and prove this result, we recall the necessary background regarding the Schouten--Nijenhuis bracket following~\cite{LZ13}. 

\subsubsection{Local multivector fields}

The space of {\it densities of local multivector fields} (on the formal loop space of $V$) is the associative graded commutative algebra
\[
\hcA\coloneqq\mbC[[u^*]][u^*_{\ge 1},\theta_{*,*}][[\eps]],
\]
where the range of the indexes of the new formal variables $\theta_{\alpha,k}$ is $1\le\alpha\le N$, $k\ge 0$. The algebra $\hcA$ is graded commutative with respect to the {\it super gradation}, denoted by $\deg_\theta$, which is defined by $\deg_\theta\theta_{\alpha,k}\coloneqq1$ and $\deg_\theta u^\alpha_k=\deg_\theta\eps\coloneqq0$. Note that the super degree $0$ homogeneous component coincides with the space of differential polynomials $\hcA^0$. The super degree $p$ homogeneous component of~$\hcA$ is denoted by $\hcA^p$

The \emph{standard gradation} is extended from $\hcA^0\subset\hcA$ to $\hcA$ by assigning $\deg\theta_{\alpha,k}\coloneqq k$. The homogeneous component of $\hcA$ of standard degree $d$ is denoted by $\hcA_d$, and we also denote
$$
\hcA^p_d\coloneqq\hcA_d\cap\hcA^p.
$$

The operator $\d_x$ is extended from $\hcA^0$ to $\hcA$ by
$$
\d_x\coloneqq\sum_{k\ge 0}u^\alpha_{k+1}\frac{\d}{\d u^\alpha_k}+\sum_{k\ge 0}\theta_{\alpha,k+1}\frac{\d}{\d\theta_{\alpha,k}}.
$$
It increases the standard degree by $1$ and preserves the super degree. Therefore, the space
$$
\hLambda\coloneqq\hcA\Big/(\mbC[[\eps]]\oplus\Im\,\d_x),
$$
called the space of {\it local multivector fields}, still possesses both gradations. The corresponding homogeneous components are denoted by $\hLambda^p_d$. If $f\in\hcA$, its image in $\hLambda$ is denoted by $\of=\int f dx$.

The \emph{variational derivatives} $\frac{\delta}{\delta u^\alpha}$ and $\frac{\delta}{\delta\theta_\alpha}$ on $\hcA$ are defined by
$$
\frac{\delta}{\delta u^\alpha}\coloneqq\sum_{i\ge 0}(-\d_x)^i\circ\frac{\d}{\d u^\alpha_i}\qquad \text{and} \qquad \frac{\delta}{\delta \theta_\alpha}\coloneqq\sum_{i\ge 0}(-\d_x)^i\circ\frac{\d}{\d \theta_{\alpha,i}},
$$ 
and, since they vanish on $\mbC[[\eps]]\oplus\Im\,\d_x$, they restrict to $\hLambda$: 
\[
\frac{\delta}{\delta u^\alpha},\frac{\delta}{\delta\theta_\alpha}\colon\hLambda\to\hcA.
\] 

There is a bijection between the space of $N$-tuples $\oQ=(Q^1,\ldots,Q^N)$, $Q^\alpha\in\hcA^0$, and the space $\hLambda^1$ (the space of \emph{local vector fields}) given by
$$
\oQ\mapsto V_\oQ\coloneqq\int Q^\alpha\theta_\alpha dx\in\hLambda^1.
$$
There is also a bijection between the space of skew-symmetric operators $K=(K^{\alpha\beta})$, $K^{\alpha\beta}=\sum_{s\ge 0}K^{\alpha\beta}_s\d_x^s$, $K^\dagger=-K$, and the space $\hLambda^2$ (the space of \emph{local bivector fields}) given by
$$
K\mapsto B_K\coloneqq\frac{1}{2}\int\sum_{s\ge 0}K^{\alpha\beta}_s\theta_{\alpha,0}\theta_{\beta,s}dx\in\hLambda^2.
$$ 

\begin{definition}
The {\it Schouten--Nijenhuis bracket}
$
[\cdot,\cdot]\colon\hLambda^p\times\hLambda^q\to\hLambda^{p+q-1}
$
is defined by
$$
[P,Q]\coloneqq\int\left(\frac{\delta P}{\delta\theta_\alpha}\frac{\delta Q}{\delta u^\alpha}+(-1)^p\frac{\delta P}{\delta u^\alpha}\frac{\delta Q}{\delta \theta_\alpha}\right)dx.
$$	
\end{definition}

\begin{lemma}[\cite{LZ11}] For any $P\in\hLambda^p$, $Q\in\hLambda^q$, $R\in\hLambda^r$ 
the Schouten--Nijenhuis bracket satisfies the properties
\begin{align*}
&[P,Q]=(-1)^{pq}[Q,P],\\
&(-1)^{pr}[[P,Q],R]+(-1)^{rq}[[R,P],Q]+(-1)^{qp}[[Q,R],P]=0.
\end{align*}
\end{lemma}

A skew-symmetric operator $K$ is \emph{Poisson} if and only if $[B_K,B_K]=0$. For a Poisson operator~$K$ and any $R\in\hLambda$ we have
\begin{gather}\label{eq:double commutator}
[[R,B_K],B_K]=0.
\end{gather}

For a skew-symmetric operator $K$, an $N$-tuple $\oQ\in(\hcA^0)^N$, and a local functional $\of\in\hLambda^0$ we have the following properties:
\begin{align}
&\frac{\delta B_K}{\delta\theta_\alpha}=\sum_{s\ge 0}K^{\alpha\beta}_s\theta_{\beta,s}\in\hcA^1;\notag\\
&[\of,B_K]=-V_{\oP},\quad P^\alpha=\sum_{s\ge 0}K^{\alpha\beta}_s\d_x^s\frac{\delta\of}{\delta u^\beta};\notag\\
\label{eq:CommutatorVQ-BK}
&[V_\oQ,B_K]=-B_{\tK},\quad\tK^{\alpha\beta}=L_\mu(Q^\alpha)\circ K^{\mu\beta}+K^{\alpha\nu}\circ L_\nu^\dagger(Q^\beta)-\sum_{p,s\ge 0}(\d_x^p Q^\gamma)\frac{\d K^{\alpha\beta}_s}{\d u^\gamma_p}\d_x^s.
\end{align}
Note that $\{\of,\og\}_{K} = [[B_K,\of],\og]$ for any $\of,\og\in\hLambda^0$.

\subsubsection{The Schouten--Nijenhuis bracket of ${K_1}$ and ${K_2}$}

The fact that the skew-symmetric operator $K_2$ is Poisson and is compatible with the operator $K_1$ is equivalent to the system of the following two equations:
$$
[B_{K_2},B_{K_2}]=0\qquad \text{and} \qquad [B_{K_2},B_{K_1}]=0.
$$

\begin{proposition}\label{proposition:Schouten bracket of K1 and K2}
We have $[B_{K_2},B_{K_1}]=0$.
\end{proposition}

\begin{proof} 
Equation~\eqref{eq:double commutator} implies that it is sufficient to represent $B_{K_2}$ as $[V_{\oR},B_{K_1}]$ for some local vector field. We do it below in Lemma~\ref{lem:BK2-as-commutator}.
\end{proof}

Introduce the \emph{higher Euler operators} $\cT_{\alpha,k}\colon\hcA^0\to\hcA^0$, $1\le\alpha\le N$, $k\ge 0$, by
$$
\cT_{\alpha,k}\coloneqq\sum_{n\ge k}{n\choose k}(-\d_x)^{n-k}\circ\frac{\d}{\d u^\alpha_n}.
$$
Clearly, $\cT_{\alpha,0}=\frac{\delta}{\delta u^\alpha}$. The operators $\cT_{\alpha,k}$ satisfy the property $\cT_{\alpha,k+1}\circ\d_x=\cT_{\alpha,k}$, $k\ge 0$ (see \cite{KMGZ},~\cite[Section~5.1]{BPS12b},~\cite[Section~2.1]{LZ11}).

Let us choose a representative $g\in\hcA^0$ for the local functional $\og\in\hLambda^0$ and define an $N$-tuple $\oR=(R^1,\ldots,R^N)\in(\hcA^0)^N$ by
$$
R^\alpha\coloneqq\eta^{\alpha\beta}\left(\left(-\frac{1}{2}-\mu_\beta\right)\frac{\delta\og}{\delta u^\beta}-\frac{1}{2}A_{\beta\gamma}u^\gamma+\d_x\cT_{\beta,1}(g)\right).
$$
Note that if we change the representative for $\og$ by $g\mapsto g+\d_x h$, $h\in\hcA^0$, then the local vector field $V_\oR$ changes as follows: $V_\oR\mapsto V_\oR-[\oh,B_{K_1}]$. Therefore, the Schouten--Nijenhuis bracket~$[V_\oR,B_{K_1}]$ doesn't depend on the choice of a representative for $\og$.
\begin{lemma} \label{lem:BK2-as-commutator}
We have $B_{K_2}=[V_\oR,B_{K_1}]$.
\end{lemma}
\begin{proof} Split $R^\alpha$ into the sum $R^\alpha=R^\alpha_1+R^\alpha_2+R^\alpha_3$, where
\[
R^\alpha_1 \coloneqq -\eta^{\alpha\beta}\left(\frac{1}{2}+\mu_\beta\right)\frac{\delta\og}{\delta u^\beta}, \qquad
R^\alpha_2 \coloneqq -\frac{1}{2}A^\alpha_{\gamma}u^\gamma, \qquad \text{and}\qquad
R^\alpha_3 \coloneqq \eta^{\alpha\beta} \d_x\cT_{\beta,1}(g).
\]
Using Equation~\eqref{eq:second expression for K2} we split $K_2$ into the sum of skew-symmetric operators, $K_2=K_{2}^{(1)}+K_{2}^{(2)}+K_{2}^{(3)}$, where
\begin{align*}
K_{2}^{(1)}& \coloneqq \d_x\circ\hOmega(\og)\circ\left(\frac{1}{2}-\mu\right)+\left(\frac{1}{2}-\mu\right)\circ\hOmega(\og)\circ\d_x, \\
K_{2}^{(2)} & \coloneqq \eta^{-1}A\eta^{-1}\d_x, \\
K_{2}^{(3)} & \coloneqq \d_x\circ\hOmega^1(\og)\circ\d_x.
\end{align*}
We prove below for $i=1,2,3$ that $B_{K_{2}^{(i)}} = [V_{\oR_i}, B_{K_1}]$.

For $i=1$ we have
\begin{align*}
K_{2}^{(1);\alpha\beta} & = \left(\d_x\circ\hOmega(\og)^\dagger\circ\left(\frac{1}{2}-\mu\right)+\left(\frac{1}{2}-\mu\right)\circ\hOmega(\og)\circ\d_x\right)^{\alpha\beta} \\
& = \eta^{\alpha\nu}\d_x\circ L_\nu^\dagger\left(\eta^{\beta\gamma}\frac{\delta\og}{\delta u^\gamma}\left(\frac{1}{2}-\mu_\beta\right)\right)+L_\nu\left(\eta^{\alpha\gamma}\frac{\delta\og}{\delta u^\gamma}\left(\frac{1}{2}-\mu_\alpha\right)\right)\circ\eta^{\nu\beta}\d_x
\\
& = \eta^{\alpha\nu}\d_x\circ L_\nu^\dagger\left(\eta^{\beta\gamma}\frac{\delta\og}{\delta u^\gamma}\left(\frac{1}{2}+\mu_\gamma\right)\right)+L_\nu\left(\eta^{\alpha\gamma}\frac{\delta\og}{\delta u^\gamma}\left(\frac{1}{2}+\mu_\gamma\right)\right)\circ\eta^{\nu\beta}\d_x.
\end{align*}
Then Equation~\eqref{eq:CommutatorVQ-BK} implies that $B_{K_2^{(1)}} = [V_{\oR_1}, B_{K_1}]$.

For $i=2$ we have
$$
K_2^{(2);\alpha\beta}=\frac{1}{2}\left(A^\alpha_\nu\circ\eta^{\nu\beta}\d_x+\eta^{\alpha\nu}\d_x\circ A_\nu^\beta\right),
$$
which implies $B_{K_2^{(2)}} = [V_{\oR_2}, B_{K_1}]$.

For $i=3$ we have (cf. the computation in the proof of Lemma~\ref{lemma:property of hOmegak})
\begin{align*}
K_{2}^{(3);\alpha\beta} & =\d_x\circ \eta^{\alpha\mu}\eta^{\beta\nu} \sum_{i=0}^\infty (i+1)   \frac{\d}{\d u^{\nu}_{i+1}}\left(\frac{\delta\og}{\delta u^\mu}\right) \d_x^{i+1} \\
& = {\eta^{\alpha\mu}\eta^{\beta\nu}} \sum_{i,k,m=0}^\infty\sum_{l=0}^{i+1} (i+1) (-1)^{l+1}\binom{m+k+l}{k,l} (-\d_x)^{m+1}\circ \left(\frac{\d^2 g}{\d u^{\nu}_{i+1-l}\d u^\mu_{m+k+l}}\right) \d_x^{i+1+k} \\
&=\eta^{\alpha\mu}\eta^{\beta\nu}\sum_{m,n\ge 0}(-\d_x)^{m+1}\circ\left((m+1)\frac{\d^2g}{\d u^\mu_{m+1}\d u^\nu_n}-(n+1)\frac{\d^2g}{\d u^\mu_m\d u^\nu_{n+1}}\right)\circ\d_x^{n+1}
\end{align*}
\begin{align*}
&=\left(\sum_{m\ge 0}(m+1)(-\d_x)^{m+1}\circ L_\nu\left(\eta^{\alpha\mu}\frac{\d g}{\d u^\mu_{m+1}}\right)\right)\circ\eta^{\nu\beta}\d_x\\
& \phantom{=\ }+\eta^{\alpha\mu}\d_x\circ\left(\sum_{n\ge 0}L_\mu^\dagger\left(\eta^{\beta\nu}\frac{\d g}{\d u^\nu_{n+1}}\right)\circ(n+1)\d_x^{n+1}\right)\\
&=L_\nu\left(-\eta^{\alpha\mu}\d_x\cT_{\mu,1}(g)\right)\circ\eta^{\nu\beta}\d_x+\eta^{\alpha\mu}\d_x\circ L_\mu^\dagger\left(-\eta^{\beta\nu}\d_x\cT_{\nu,1}(g)\right),
\end{align*}
which implies $B_{K_2^{(3)}} = [V_{\oR_3}, B_{K_1}]$. In this computation we pass from the second line to the third one using the following straightforward combinatorial identity:
\begin{align*}
&  \sum_{\substack{i+k=n\\ k+l=p}}  (i+1) (-1)^{l+1}\binom{m+k+l}{k,l} \\
& \hspace{1cm}
= -(n+1) \sum_{k+l=p}   (-1)^{l}\binom{m+k+l}{k,l} + \sum_{k+l=p} k (-1)^{l}\binom{m+k+l}{k,l}
\\ 
& \hspace{1cm} = -(n+1)\delta_{p,0} + (m+1)\delta_{p,1}.
\end{align*}
\end{proof}


\section{Evidence II: central invariants}\label{section:central invariants}

In this section, assuming that the first part of Conjecture~\ref{main conjecture}, is true we compute the central invariants of the pair of operators~$(K_1,K_2)$. We also discuss the relation between Conjecture~\ref{main conjecture}, the conjecture on the equivalence of the DR hierarchy and the Dubrovin--Zhang hierarchy, and the Dubrovin--Zhang conjecture on existence of a bihamiltonian structure for the Dubrovin--Zhang hierarchy. 

\subsection{Miura transformations}

Let us discuss changes of coordinates on the formal loop space. In order to distinguish rings of differential polynomials in different variables, let ~$\cA^0_w$ denote the ring of differential polynomials in variables~$w^1,\ldots,w^N$. The same notation is adopted for the extension $\hcA^0_w$ and for the spaces of local functionals~$\Lambda^0_w$ and~$\hLambda^0_w$. 

\begin{definition}
Changes of variables of the form
\begin{gather}\label{eq:Miura transformation}
u^\alpha\mapsto\tu^\alpha(u^*_*,\eps)=\sum_{k\ge 0}\eps^k f^\alpha_k(u^*_*),\quad 1\le\alpha\le N,\quad\text{where}\quad f^\alpha_k\in\cA^0_{u;k},\quad \left.\det\left(\frac{\d f_0^\alpha}{\d u^\beta}\right)\right|_{u^*=0}\ne 0,
\end{gather}
are called {\it Miura transformations}. 
\end{definition}

We say that a Miura transformation is {\it close to identity} if $f^\alpha_0=u^\alpha$. 

Under a Miura transformation any differential polynomial $P(u^*_*,\eps)\in\hcA^0_{u}$ can be rewritten as a formal power series in $\eps$ whose coefficients are polynomials in the variables~$\tu^\alpha_i$, $i>0$, with coefficients in the ring of formal power series in the shifted variables $\tu^\alpha-\tu^\alpha_{\orig}$, where $\tu^\alpha_\orig\coloneqq\left.\tu^\alpha(u^*_*,\eps)\right|_{u^*_*=\eps=0}=\left.f^\alpha_0\right|_{u^*=0}$. We introduce the notation
$$
\hcA^{\otu_\orig;0}_{\tu}\coloneqq\mbC[[\tu^\alpha-\tu^\alpha_\orig]][\tu^\alpha_{\ge 1}][[\eps]],\quad\otu_\orig=(\tu^1_\orig,\ldots,\tu^N_\orig),
$$ 
and still call the  elements of the ring $\hcA^{\otu_\orig;0}_{\tu}$ differential polynomials. Let~$\hLambda^{\otu_\orig;0}_{\tu}$ denote the corresponding space of local functionals. A differential polynomial $P(u^*_*,\eps)\in\hcA^0_{u}$ rewritten in the variables $\tu^\alpha$ is denoted by $P(\tu^*_*,\eps)\in\hcA^{\otu_\orig;0}_{\tu}$. The condition~$f^\alpha_k\in\cA^0_{u;k}$ guarantees that if $P(u^*_*,\eps)\in\hcA^0_{u;d}$, then $P(\tu^*_*,\eps)\in\hcA^{\otu_\orig;0}_{\tu;d}$. This means that a Miura transformation defines an isomorphism $\hcA^0_{u;d}\simeq\hcA^{\otu_\orig;0}_{\tu;d}$. It also induces an isomorphism $\hLambda^0_{u;d}\simeq\hLambda^{\otu_\orig;0}_{\tu;d}$. The image of a local functional $\oh[u]\in\hLambda^0_{u;d}$ under this isomorphism is denoted by $\oh[\tu]\in\hLambda^{\otu_\orig;0}_{\tu;d}$. 

Let us describe the action of Miura transformations on Hamiltonian hierarchies. Consider a Hamiltonian hierarchy~\eqref{eq:Hamiltonian system} and a Miura transformation~\eqref{eq:Miura transformation}. Then in the new variables~$\tu^\alpha$ the system~\eqref{eq:Hamiltonian system} reads:
\begin{align*}
&\frac{\d\tu^\alpha}{\d\tau_i}=K_{\tu}^{\alpha\mu}\frac{\delta\oh_i[\tu]}{\delta \tu^\mu},\qquad
K_{\tu}^{\alpha\beta}=L_\mu(\tu^\alpha(u^*_*,\eps))\circ K^{\mu\nu}\circ L_\nu^\dagger(\tu^\beta(u^*_*,\eps)).
\end{align*}

\subsection{The central invariants of a pair of compatible Poisson operators}

Consider a pair of compatible Poisson operators $(P_1,P_2)$ of degree~$1$:
$$
P^{\alpha\beta}_a=g^{\alpha\beta}_a\d_x+\Gamma^{\alpha\beta}_{a;\gamma}u^\gamma_x+O(\eps),\quad g^{\alpha\beta}_a,\Gamma^{\alpha\beta}_{a;\gamma}\in\mbC[[u^*]],\quad a=1,2.
$$
Assume that 
$$
\left.\det(g^{\alpha\beta}_1)\right|_{u^*=0}\ne 0
$$
and that there exist formal power series $\hu^i(u^*)\in\mbC[[u^*]]$, $1\le i\le N$, such that
\begin{align}
\det\left(g^{\alpha\beta}_2-\hu^i g^{\alpha\beta}_1\right)=&0,\quad 1\le i\le N,\notag\\
\left.\det\left(\frac{\d\hu^i}{\d u^\alpha}\right)\right|_{u^*=0}\ne&0.\label{eq:nondegeneracy for hu}
\end{align}
Clearly, if such an $N$-tuple $(\hu^1,\ldots,\hu^N)$ exists, then it is unique up to permutations. Condition~\eqref{eq:nondegeneracy for hu} implies that the functions $\hu^i(u^*)$ can be used as a system of coordinates. These coordinates are called the {\it canonical coordinates} of the pair of operators $(P_1,P_2)$. 

Let us use Latin superscripts to indicate components of the operators $P_1,P_2$ in canonical coordinates, $P_a^{ij}\coloneqq P^{ij}_{a;\hu}$. The operators $P_1$ and $P_2$ have the following form in canonical coordinates (see e.g.~\cite{DLZ06}):
\begin{align}
&P^{ij}_a=\sum_{k\ge 0}\eps^k\sum_{l=0}^{k+1}P^{[k],ij}_{a;l}\d_x^l,\quad\text{where}\label{eq:form for central invariants,1}\\
&P^{[0],ij}_{1;1}=\delta_{ij}f^i,\quad P^{[0],ij}_{2;1}=\delta_{ij}\hu^i f^i,\quad f^i\in\mbC[[\hu^j-\hu^j_\orig]],\quad\hu^j_\orig=\left.\hu^j(u^*)\right|_{u^*=0}.\label{eq:form for central invariants,2}
\end{align}
The {\it central invariants} of the pair of operators $(P_1,P_2)$ are the formal power series $c_i(\hu^i)\in\mbC[[\hu^i-\hu^i_\orig]]$, $1\le i\le N$, defined by (see e.g.~\cite{DLZ06})
\begin{gather}\label{eq:definition of central invariants}
c_i(\hu^i)\coloneqq\frac{1}{3(f^i)^2}\left(P^{[2],ii}_{2;3}-\hu^i P^{[2],ii}_{1;3}+\sum_{k\ne i}\frac{\left(P^{[1],ki}_{2;2}-\hu^i P^{[1],ki}_{1;2}\right)^2}{(\hu^k-\hu^i)f^k}\right).
\end{gather}
\begin{remark}
In~\cite{DLZ06} the proof that the right-hand side of this formula is a formal power series that doesn't depend on $\hu^j$ with $j\ne i$ uses the assumption that $\hu^j_\orig\ne 0$ for all $j$ and $\hu^j_\orig\ne\hu^k_\orig$ for all $j\ne k$. Note, however, that the same argument allows to prove that $\frac{\d c_i}{\d \hu^j}=0$ for $j\ne i$ without this assumption. Also, Equation~\eqref{eq:definition of central invariants} and the fact that $f^i$ doesn't vanish at the origin imply that $c_i=\frac{g_i(\hu^1,\ldots,\hu^N)}{\prod_{j\ne i}(\hu^j-\hu^i)}$ for some $g_i(\hu^1,\ldots,\hu^N)\in\mbC[[\hu^1-\hu^1_\orig,\ldots,\hu^N-\hu^N_\orig]]$. For any fixed $j\ne i$, the fact that $\frac{\d c_i}{\d\hu^j}=0$ implies that $g_i=(\hu^j-\hu^i)\frac{\d g_i}{\d \hu^j}$. Therefore, $g_i$ is divisible by $\hu^j-\hu^i$ for any $j\ne i$ in the ring $\mbC[[\hu^1-\hu^1_\orig,\ldots,\hu^N-\hu^N_\orig]]$. Thus, the quotient $\frac{g_i(\hu^1,\ldots,\hu^N)}{\prod_{j\ne i}(\hu^j-\hu^i)}$ belongs to the ring $\mbC[[\hu^i-\hu^i_\orig]]$.
\end{remark}
The central invariants of the pair $(P_1,P_2)$ are invariant under Miura transformations of the variables $\hu^i$ close to identity. Moreover, in case $\hu^i_\orig\ne 0$ and $\hu^i_\orig\ne\hu^j_\orig$ for $i\ne j$, the central invariants classify pairs of compatible Poisson operators of the form~\eqref{eq:form for central invariants,1}--\eqref{eq:form for central invariants,2} with fixed functions $f^i$ under Miura transformations of the variables $\hu^i$ close to identity (see~\cite{LZ05,DLZ06,LZ13,CPS18,CKS18}).

\subsection{The central invariants of the pair $(K_1,K_2)$}   
Consider a homogeneous semisimple CohFT. The associated formal Dubrovin--Frobenius manifold is semisimple at the origin, and therefore it has canonical coordinates $\hu^i(t^*)\in\mbC[[t^*]]$, $1\le i\le N$ (see e.g.~\cite[Theorem~3.3]{Man99}):
$$
\frac{\d}{\d\hu^i}\circ\frac{\d}{\d\hu^j}=\delta_{ij}\frac{\d}{\d\hu^i},
$$
where by $\circ$ we denote here the multiplication in the tangent spaces of the Dubrovin--Frobenius manifold. After an appropriate shift $\hu^i\mapsto\hu^i+a^i$, $a^i\in\mbC$, the Euler vector field $E$ written in canonical coordinates is given by (see e.g.~\cite[Theorem~3.6]{Man99})
$$
E=\sum_{i=1}^N\hu^i\frac{\d}{\d\hu^i}.
$$
Moreover, $\det\left(g^{\alpha\beta}-\hu^i\eta^{\alpha\beta}\right)=0$ for $1\le i\le N$, where $g^{\alpha\beta}$ is given by Equation~\eqref{eq:galphabeta} and we identify $t^\gamma=u^\gamma$. 

Therefore, if we assume that the operator $K_2$ is Poisson, then by Proposition~\ref{proposition:Schouten bracket of K1 and K2} the operators~$K_1$ and~$K_2$ form a pair of compatible Poisson operators and we can compute the associated central invariants.

\begin{proposition}\label{proposition:central invariants of K1 and K2}
Suppose that the operator $K_2$ corresponding to our homogeneous semisimple CohFT is Poisson. Then all the central invariants of the pair of operators $(K_1,K_2)$ are equal to $\frac{1}{24}$.
\end{proposition}
\begin{proof}
In canonical coordinates the operators $K_1$ and $K_2$ have the form
\begin{align*}
K_1^{ij}=&K_{1;1}^{[0],ij}\d_x+K_{1;0}^{[0],ij}, && K_{1;1}^{[0],ij}=\delta_{ij}f^i, && f^i\in\mbC[[\hu^j-\hu^j_\orig]],\\
K_2^{ij}=&K_{2;1}^{[0],ij}\d_x+K_{2;0}^{[0],ij}+\sum_{k\ge 1}\eps^{2k}\sum_{l=0}^{2k+1}K^{[2k],ij}_{2;l}\d_x^l, && K_{2;1}^{[0],ij}=\delta_{ij}\hu^i f^i.
\end{align*}
Therefore,
$$
c_i=\frac{K^{[2],ii}_{2;3}}{3(f^i)^2}.
$$	
Consider the expansion
$$
\og=\sum_{g\ge 0}\og^{[2g]}\eps^{2g},\quad \og^{[2g]}\in\Lambda^0_{2g}.
$$
We have
\begin{align*}
K^{[2],\alpha\beta}_2=&\eta^{\alpha\mu}\eta^{\beta\nu}\left(\left(\frac{1}{2}-\mu_\beta\right)\d_x\circ L_\nu\left(\frac{\delta\og^{[2]}}{\delta u^\mu}\right)+\left(\frac{1}{2}-\mu_\alpha\right)L_\nu\left(\frac{\delta\og^{[2]}}{\delta u^\mu}\right)\circ\d_x\right. \\
&\left.+\d_x\circ L^1_\nu\left(\frac{\delta\og^{[2]}}{\delta u^\mu}\right)\circ\d_x\right),
\end{align*}
which implies that
$$
K^{[2],\alpha\beta}_{2;3}=\left(3-\mu_\alpha-\mu_\beta\right)\eta^{\alpha\mu}\eta^{\beta\nu}\frac{\d}{\d u^\nu_2}\frac{\delta\og^{[2]}}{\delta u^\mu}.
$$
By~\eqref{eq:og up to genus 1}, we have
$$
\og^{[2]}=-\frac{1}{48}\int c^\theta_{\theta\xi}c^\xi_{\alpha\beta}u^\alpha_xu^\beta_x dx,
$$
where we recall that $c^\alpha_{\beta\gamma}=\eta^{\alpha\mu}\frac{\d^3F}{\d t^\mu\d t^\beta\d t^\gamma}$ and we use the identification  $t^\theta=u^\theta$. Therefore, $\frac{\d}{\d u^\nu_2}\frac{\delta\og^{[2]}}{\d u^\mu}=\frac{1}{24}c^\theta_{\theta\xi}c^\xi_{\mu\nu}$ and, hence,
$$
K^{[2],\alpha\beta}_{2;3}=\frac{1}{24}(3-\mu_\alpha-\mu_\beta)c^\theta_{\theta\xi}c^{\xi\alpha\beta},
$$
where we raise the indices in the tensor $c^\alpha_{\beta\gamma}$ using the metric $\eta$. 
	
Note that 
$$
K^{[2],ij}_{2;3}=\frac{\d\hu^i}{\d u^\alpha}\frac{\d\hu^j}{\d u^\beta}K^{[2],\alpha\beta}_{2;3},
$$
which means that the collection of functions $K^{[2],\alpha\beta}_{2;3}$ is transformed to the canonical coordinates as a tensor. We use Latin indices for the expressions of the operator $\mu$, the metric $\eta$, and the tensor~$c^\alpha_{\beta\gamma}$ in canonical coordinates. We have
$\eta^{ij}=\delta_{ij}f^i$, $c^i_{jk}=\delta_{ij}\delta_{jk}$, and $c^{ijk}=(f^i)^2\delta_{ij}\delta_{jk}$.  
Therefore, 
$$
K^{[2],ij}_{2;3}=\frac{1}{24}\left(3\delta_{ij}(f^i)^2-\mu^i_j(f^j)^2-(f^i)^2\mu^j_i\right).
$$
Equation~\eqref{eq:mu and eta} implies that
$\mu^i_j f^j+f^i\mu^j_i=0$, hence $\mu^i_i=0$.
We conclude that
$K^{[2],ii}_{2;3}=\frac{(f^i)^2}{8}$, hence $c_i=\frac{1}{24}$.
\end{proof}

\subsection{Relation with the Dubrovin--Zhang hierarchy} \label{sec:Dubrovin-Zhang}

Beside the DR hierarchy, which we discuss in this paper, there is another Hamiltonian hierarchy associated to a semisimple homogeneous CohFT, called the {\it Dubrovin--Zhang (DZ) hierarchy} or the {\it hierarchy of topological type}~\cite{DZ01,BPS12b}. It has the same dispersionless part as the DR hierarchy. By a conjecture in~\cite{Bur15}, called the {\it DR/DZ equivalence conjecture} (see also a stronger version in~\cite{BDGR18}), it is related to the DR hierarchy by a Miura transformation that is close to identity. A long-standing open conjecture proposed by Dubrovin and Zhang claims that the DZ hierarchy possesses a bihamiltonian structure. 

We see that the Dubrovin--Zhang conjecture follows from the DR/DZ equivalence conjecture and Conjecture~\ref{main conjecture}. The Dubrovin--Zhang conjecture is checked at the approximation up to~$\eps^2$~\cite{DZ98}, and the central invariants of the resulting pair of compatible Poisson operators are all equal to $\frac{1}{24}$ (see e.g.~\cite{Liu18}), which agrees with Proposition~\ref{proposition:central invariants of K1 and K2}.


\section{Evidence III: explicit examples}\label{section: examples}

In this section we prove Conjecture~\ref{main conjecture} in several special cases, for particular semisimple homogeneous CohFTs.

\subsection{The trivial cohomological field theory}

Consider the trivial CohFT with $N=1$, $V=\<e_1\>$, $e=e_1$, $\eta_{1,1}=1$, and $c^\triv_{g,n}(e_1^{\otimes n})=1\in H^*(\oM_{g,n})$. Let $u_n$ denote $u^1_n$. The corresponding DR hierarchy coincides with the KdV hierarchy, 
$$
\og_{1,d}=\oh^{\KdV}_d, \quad d\ge -1
$$ 
(see \cite[Section 4.3.1]{Bur15}), where the sequence of Hamiltonians $\oh^\KdV_d\in\hLambda^0_0$, $d\ge -1$, of the KdV hierarchy is uniquely determined by the properties 
\begin{align*}
\oh^\KdV_1=&\int\left(\frac{u^3}{6}+\eps^2\frac{uu_{xx}}{24}\right)dx,\\
\oh^\KdV_d=&\int\left(\frac{u^{d+2}}{(d+2)!}+O(\eps^2)\right)dx,\quad d\ge -1,
\end{align*}
and the commutation relations
$$
\left\{\oh^\KdV_d,\oh^\KdV_1\right\}_{\d_x}=0,\quad d\ge -1,
$$
with respect to the Poisson bracket $\{\cdot,\cdot\}_{\d_x}$.

The trivial CohFT is homogeneous with 
$q_1=0$, $\delta=0$, $A=0$, $\mu_1=0$.
We have
$$
\og=\int\left(\frac{u^3}{6}+\eps^2\frac{uu_{xx}}{48}\right)dx,\qquad\frac{\delta\og}{\delta u}=\frac{u^2}{2}+\eps^2\frac{u_{xx}}{24},\qquad \hOmega(\og)=u+\frac{\eps^2}{24}\d_x^2,\qquad \hOmega^1(\og)=\frac{\eps^2}{12}\d_x,
$$
which, using formula~\eqref{eq:second expression for K2}, gives
$$
K_2=u\d_x+\frac{1}{2}u_x+\frac{\eps^2}{8}\d_x^3.
$$
The fact that the pair $(\d_x,u\d_x+\frac{1}{2}u_x+\frac{\eps^2}{8}\d_x^3)$ is a pair of compatible Poisson operators and that the Hamiltonians $\oh^\KdV_d$ satisfy the relations
$$
\left\{\cdot,\oh^\KdV_d\right\}_{K_2}=\left(d+\frac{3}{2}\right)\left\{\cdot,\oh^\KdV_{d+1}\right\}_{\d_x},\quad d\ge -1,
$$
is a standard fact in the theory of the KdV hierarchy (see e.g.~\cite{Dic03}). Thus, Conjecture~\ref{main conjecture} holds for the trivial CohFT.

\subsection{Higher $r$-spin theories}\label{subsubsection:higher rspin}
Let $r\ge 2$ and consider an $(r-1)$-dimensional complex vector space~$V$ with a fixed basis $e_1,\ldots,e_{r-1}$. Let $\eta_{\alpha\beta}\coloneqq\delta_{\alpha+\beta,r}$. There exists a unique CohFT (see e.g.~\cite{PPZ15}) $c^\rspin_{g,n}\colon V^{\otimes n} \to H^\even(\oM_{g,n})$ satisfying the following properties:
\begin{enumerate}
	\item $c^\rspin_{g,n}(\otimes_{i=1}^n e_{\alpha_i})=0$ if $r \mathop{\not|\ } (g-1)(r-2)+\sum_{i=1}^n(\alpha_i-1)$;
	\item $\deg c^\rspin_{g,n}(\otimes_{i=1}^n e_{\alpha_i})=2\tfrac{(g-1)(r-2)+\sum_{i=1}^n(\alpha_i-1)}{r}$ if $r  \mathop{|\ } (g-1)(r-2)+\sum_{i=1}^n(\alpha_i-1)$;
	\item $c^\rspin_{0,3}(e_\alpha\otimes e_\beta\otimes e_\gamma)=\delta_{\alpha+\beta+\gamma,r+1}\in H^0(\oM_{0,3})\cong \mbC$;
	\item $c^\rspin_{0,4}(e_2^{\otimes 2}\otimes e_{r-1}^{\otimes 2})=\frac{1}{r}[\pt]\in H^2(\oM_{0,4})$ if $r\ge 3$.
\end{enumerate}
Here $[\pt]$ denotes the cohomology class dual to a point in $\oM_{0,4}$. This CohFT is called the {\it $r$-spin CohFT}. It is homogeneous with
$q_\alpha=\frac{\alpha-1}{r}$, $r^\alpha=0$, $\delta=\frac{r-2}{r}$, $\mu_\alpha=\frac{2\alpha-r}{2r}$, $A_{\alpha\beta}=0$.
The $2$-spin CohFT coincides with the trivial CohFT.

In this section we check Conjecture~\ref{main conjecture} for the $r$-spin CohFT for $r=3,4,5$. To this end, we first describe the construction of the Gelfand--Dickey hierarchy, then we present its relation (through the Dubrovin--Zhang hierarchy) to the DR hierarchy for the $r$-spin theory for $r=3,4,5$ proved in~\cite{BG16}, and then we finally check Conjecture~\ref{main conjecture} for the $r$-spin CohFT for these values of~$r$. 

\begin{remark}
Note that for any differential polynomial $f\in\cA^0$ there exists a unique differential polynomial $\widetilde{f}\in\hcA^0_0$ such that $\left.\widetilde{f}\right|_{\eps=0}=f$. This gives a natural inclusion $\cA^0\hookrightarrow\hcA^0_0$. The same remark is true for local functionals. Regarding Poisson brackets, note that any Poisson operator~$K$, $K^{\alpha\beta}=\sum_{i\ge 0}K^{\alpha\beta}_i\d_x^i$, $K^{\alpha\beta}_i\in\cA^0$, can be expressed in a unique way as the sum of homogeneous operators: $K^{\alpha\beta}=\sum_{d\ge 0}\left(\sum_{i+j=d}K^{\alpha\beta}_{i,j}\d_x^i\right)$, where $K^{\alpha\beta}_{i,j}\in\cA^0_{d-i}$. If its degree zero part vanishes, i.e., if $K^{\alpha\beta}_{0,0}=0$, then we can naturally associate to such a Poisson operator $K$ the Poisson operator $\widetilde{K}^{\alpha\beta}=\sum_{d\ge 1}\eps^{d-1}\left(\sum_{i+j=d}K^{\alpha\beta}_{i,j}\d_x^i\right)$ that gives a Poisson bracket of degree~$1$ on the space $\hLambda^0$. We silently use this correspondence throughout this section.
\end{remark}

\subsubsection{Gelfand--Dickey hierarchy}

We review the construction of the Gelfand--Dickey hierarchy and its bihamiltonian structure following~\cite{Dic03}. 

Consider variables~$f_0,f_1,\ldots,f_{r-2}$. A pseudo-differential operator $A$ is a Laurent series
$$
A=\sum_{n=-\infty}^m a_n\d_x^n,
$$
where $m$ is an arbitrary integer and $a_n\in\cA^0_{f_0,f_1,\ldots,f_{r-2}}$. Let
\begin{gather*}
A_+\coloneqq\sum_{n=0}^m a_n\d_x^n,\qquad \res A\coloneqq a_{-1}.
\end{gather*}
The product of pseudo-differential operators is defined by the following commutation rule:
\begin{gather*}
\d_x^k\circ a\coloneqq\sum_{l=0}^\infty\frac{k(k-1)\ldots(k-l+1)}{l!}(\d_x^l a)\d_x^{k-l},\quad a\in\cA^0_{f_0,f_1,\ldots,f_{r-2}},\quad k\in\mbZ.
\end{gather*}
For any $m\ge 2$ and a pseudo-differential operator~$A$ of the form
$$
A=\d_x^m+\sum_{n=1}^\infty a_n\d_x^{m-n}
$$
there exists a unique pseudo-differential operator $A^{\frac{1}{m}}$ of the form
$$
A^{\frac{1}{m}}=\d_x+\sum_{n=0}^\infty \widetilde{a}_n\d_x^{-n}
$$
such that $\left(A^{\frac{1}{m}}\right)^m=A$.
Let 
$$
L\coloneqq\d_x^r+f_{r-2}\d_x^{r-2}+\ldots+f_1\d_x+f_0.
$$
The $r$-th Gelfand--Dickey (GD) hierarchy is the following system of partial differential equations: 
\begin{gather}\label{eq:GD hierarchy}
\frac{\d L}{\d T^\alpha_a}=\left[\left(L^{a+\frac{\alpha}{r}}\right)_+,L\right],\quad 1\le\alpha\le r-1,\quad a\ge 0.
\end{gather}

Let us describe a Hamiltonian structure for the GD hierarchy. Let $X_0,X_1,\ldots,X_{r-2}\in\cA^0_{f_0,\ldots,f_{r-2}}$ be some differential polynomials. Consider a pseudo-differential operator 
$$
X\coloneqq\d_x^{-(r-1)}\circ X_{r-2}+\ldots+\d_x^{-1}\circ X_0.
$$ 
It is easy to see that the operator $[X,L]_+$ has the form
$$
[X,L]_+=\sum_{0\le\alpha,\beta\le r-2}\left(K^{\GD;\alpha\beta}_1X_\beta\right)\d_x^\alpha,
$$
where
$$
K^{\GD;\alpha\beta}_1=\sum_{i\ge 0}K^{\GD;\alpha\beta}_{1;i}\d_x^i,\qquad K^{\GD;\alpha\beta}_{1;i}\in\cA^0_{f_0,\ldots,f_{r-2}},
$$
and the sum is finite. The operator $K^{\GD}_1=\left(K^{\GD;\alpha\beta}_1\right)_{0\le\alpha,\beta\le r-2}$ is Poisson and its degree zero part vanishes.

Consider the local functionals
$$
\oh_{\alpha,a}^{\GD}\coloneqq-\frac{r}{(a+1)r+\alpha}\int\res L^{a+1+\frac{\alpha}{r}}dx,\quad 1\le\alpha\le r-1,\quad a\ge -1.
$$
We have 
$$
\left\{\oh^{\GD}_{\alpha,a},\oh^{\GD}_{\beta,b}\right\}_{K^{\GD}_1}=0,\quad 1\le\alpha,\beta\le r-1,\quad a,b\ge 0,
$$
and the local functionals $\oh^\GD_{\alpha,-1}$, $1\le\alpha\le N$, are linearly independent Casimirs of the Poisson bracket~$\{\cdot,\cdot\}_{K_1^\GD}$. For a local functional $\oh\in\Lambda^0_{f_0,f_1,\ldots,f_{r-2}}$ define a pseudo-differential operator~$\frac{\delta\oh}{\delta L}$ by
$$
\frac{\delta\oh}{\delta L}\coloneqq\d_x^{-(r-1)}\circ\frac{\delta\oh}{\delta f_{r-2}}+\ldots+\d_x^{-1}\circ\frac{\delta\oh}{\delta f_0}.
$$
Then the right-hand side of~\eqref{eq:GD hierarchy} can be written in the following way: 
$$
\left[\left(L^{a+\frac{\alpha}{r}}\right)_+,L\right]=\left[\frac{\delta\oh^{\GD}_{\alpha,a}}{\delta L},L\right]_+=\sum_{0\le\beta,\gamma\le r-2}\left(K^{\GD;\beta\gamma}_1\frac{\delta\oh^{\GD}_{\alpha,a}}{\delta f_\gamma}\right)\d_x^\beta.
$$
Therefore, the sequence of local functionals $\oh^{\GD}_{\alpha,a}$ together with the Poisson operator~$K^{\GD}_1$ defines a Hamiltonian structure of the GD hierarchy~\eqref{eq:GD hierarchy}.

The GD hierarchy is endowed with a bihamiltonian structure, which can be described in the following way. Let
$$
\tX\coloneqq\d_x^{-r}\circ X_{r-1}+\d_x^{-(r-1)}\circ X_{r-2}+\ldots+\d_x^{-1}\circ X_0.
$$
It is easy to see that the operator $\left(L\circ\tX\right)_+\circ L-L\circ\left(\tX\circ L\right)_+$ has the form
$$
\left(L\circ\tX\right)_+\circ L-L\circ\left(\tX\circ L\right)_+=\sum_{0\le\alpha,\beta\le r-1}\left(\tK^{\alpha\beta}X_\beta\right)\d_x^\alpha,
$$
where
$$
\tK^{\alpha\beta}=\sum_{i\ge 0}\tK^{\alpha\beta}_i\d_x^i,\quad \tK^{\alpha\beta}_i\in\cA^0_{f_0,\ldots,f_{r-2}},
$$
and the sum is finite. There is the following result:
$$
\sum_{\beta=0}^{r-1}\tK^{r-1,\beta}X_\beta=\res[\tX,L]=\d_x\left(\sum_{\substack{0\le j\le r-1\\1\le\alpha\le r-j}}{-j-1\choose\alpha}\d_x^{\alpha-1}(f_{j+\alpha}X_j)\right),
$$
where we adopt the conventions $f_r\coloneqq1$ and $f_{r-1}\coloneqq0$. Let 
$$
f(X_0,X_1,\ldots,X_{r-2})\coloneqq\frac{1}{r}\left(\sum_{\substack{0\le j\le r-2\\1\le\alpha\le r-j}}{-j-1\choose\alpha}\d_x^{\alpha-1}(f_{j+\alpha}X_j)\right).
$$
Define an operator $K^\GD_2=\left(K^{\GD;\alpha\beta}_2\right)_{0\le\alpha,\beta\le r-2}$ by the equation
$$
\sum_{0\le\alpha,\beta\le r-2}\left(\tK^{\alpha\beta} X_\beta\right)\d_x^\alpha+\sum_{\alpha=0}^{r-2}\left(\tK^{\alpha,r-1}f(X_0,\ldots,X_{r-2})\right)\d_x^\alpha=\sum_{0\le\alpha,\beta\le r-2}\left(K^{\GD;\alpha\beta}_2X_\beta\right)\d_x^\alpha.
$$
The operator $K^\GD_2$ is Poisson and compatible with the operator~$K^\GD_1$. It is easy to see that the degree zero part of the operator $K^\GD_2$ vanishes.

Consider the local functionals
$$
\og^\GD_{\alpha,a}\coloneqq\frac{r}{ar+\alpha}\int\res L^{a+\frac{\alpha}{r}}dx=-\oh_{\alpha,a-1}^\GD,\quad 1\le\alpha\le r-1,\quad a\ge 0.
$$
The right-hand side of~\eqref{eq:GD hierarchy} can be written in the following way using the local functionals $\og^\GD_{\alpha,a}$ and the Poisson operator $K_2^\GD$:
$$
\left[\left(L^{a+\frac{\alpha}{r}}\right)_+,L\right]=\sum_{0\le\beta,\gamma\le r-2}\left(K^{\GD;\beta\gamma}_2\frac{\delta\og^{\GD}_{\alpha,a}}{\delta f_\gamma}\right)\d_x^\beta,
$$
or, equivalently,
$$
\left\{\cdot,\oh^\GD_{\alpha,a}\right\}_{K_2^\GD}=-\left\{\cdot,\oh^\GD_{\alpha,a+1}\right\}_{K_1^\GD},\quad 1\le\alpha\le r-1,\quad a\ge -1,
$$
which endows the GD hierarchy with a bihamiltonian structure.

\subsubsection{Relation with the DR hierarchy}
As it is mentioned in Section~\ref{sec:Dubrovin-Zhang}, there is another bihamiltonian hierarchy associated, in particular, to the $r$-spin CohFT, the Dubrovin--Zhang (DZ) hierarchy. In this section, we recall the relation between the $r$-th GD hierarchy and the DZ hierarchy for the $r$-spin CohFT, and then present a relation between the DR hierarchy and the DZ hierarchy for $r=3$, $4$, and $5$.  

Introduce new variables $w^1,\ldots,w^{r-1}$ and consider the following Miura transformation:
$$
f_i\mapsto w^\alpha(f_{*,*})=\frac{1}{(r-\alpha)(-r)^{\frac{r-\alpha-1}{2}}}\res L^{(r-\alpha)/r}.
$$
Define Poisson operators $K_1^{\text{$r$-spin}}=\left(K_1^{\text{$r$-spin};\alpha\beta}\right)_{1\le\alpha,\beta\le r-1}$, $K_2^{\text{$r$-spin}}=\left(K_2^{\text{$r$-spin};\alpha\beta}\right)_{1\le\alpha,\beta\le r-1}$, and local functionals $\oh^{\text{$r$-spin}}_{\alpha,d}\in\Lambda_{w^1,\ldots,w^{r-1}}$, $1\le\alpha\le r-1$, $d\ge -1$, by
\begin{align*}
K_1^{\text{$r$-spin}}\coloneqq&(-r)^{\frac{r}{2}}K^{\GD}_{1;w},\\
K_2^{\text{$r$-spin}}\coloneqq&K^{\GD}_{2;w},\\
\oh_{\alpha,d}^{\text{$r$-spin}}\coloneqq&\frac{1}{(-r)^{\frac{\alpha-1+r(d+1)}{2}-d}(\alpha+rd)!_r}\oh_{\alpha,d}^{\GD}[w],
\end{align*}
where 
$$
(\alpha+rd)!_r\coloneqq\begin{cases}
\prod_{i=0}^d(\alpha+ri),&\text{if $d\ge 0$},\\
1,&\text{if $d=-1$}.
\end{cases}
$$ 
Recall that $K^\GD_{1;w}$ denotes the Miura transform of the operator~$K^{\GD}_1$. The local functionals $\oh_{\alpha,d}^{\text{$r$-spin}}$ together with the Poisson operator $K_1^{\text{$r$-spin}}$ define the Hamiltonian hierarchy
$$
\frac{\d w^\alpha}{\d t^\beta_b}=K_1^{\text{$r$-spin};\alpha\mu}\frac{\delta\oh_{\beta,b}^{\text{$r$-spin}}}{\delta w^\mu},\quad 1\le\alpha,\beta\le N,\quad b\ge 0,
$$
which is bihamiltonian with the bihamiltonian recursion given by
\begin{gather}\label{eq:bihamiltonian recursion for DZ-rspin}
\left\{\cdot,\oh^\rspin_{\alpha,d}\right\}_{K^\rspin_2}=\frac{\alpha+(d+1)r}{r}\left\{\cdot,\oh^\rspin_{\alpha,d+1}\right\}_{K^\rspin_1},\quad 1\le\alpha\le r-1,\quad d\ge -1.
\end{gather}
This bihamiltonian hierarchy is the Dubrovin--Zhang hierarchy corresponding to the $r$-spin CohFT~\cite{DZ01,FSZ10}.

In~\cite{BG16} it is proved that for $r=3$, $4$, and $5$ the local functionals $\og_{\alpha,d}$ and the Poisson operator $K_1=\eta^{-1}\d_x$ of the DR hierarchy are related to the local functionals $\oh_{\alpha,d}^\rspin$ and the Poisson operator $K_1^\rspin$ by the following Miura transformations:
\begin{align}
&\left\{
\begin{aligned}
&w^1=u^1,\\
&w^2=u^2,
\end{aligned}\right.&&\text{for $r=3$};\label{eq:Miura transformation for 3-spin}\\
&\left\{
\begin{aligned}
&w^1=u^1+\frac{\eps^2}{96}u^3_{xx},\\
&w^2=u^2,\\
&w^3=u^3,
\end{aligned}\right.&&\text{for $r=4$};\label{eq:Miura transformation for 4-spin}\\
&\left\{
\begin{aligned}
&w^1=u^1+\frac{\eps^2}{60}u^3_{xx},\\
&w^2=u^2+\frac{\eps^2}{60}u^4_{xx},\\
&w^3=u^3,\\
&w^4=u^4,
\end{aligned}\right.&&\text{for $r=5$}.\label{eq:Miura transformation for 5-spin}
\end{align}

\subsubsection{Proof of Conjecture~\ref{main conjecture} for $r=3,4,5$}

For the first part of the conjecture it is sufficient to check that the operators $K_2$ and $K_2^\rspin$ are related by the Miura transformations given by equations~\eqref{eq:Miura transformation for 3-spin},~\eqref{eq:Miura transformation for 4-spin}, and~\eqref{eq:Miura transformation for 5-spin}.

For $r=3$ we have the following formula for $K^{\text{$3$-spin}}_2$:
\begin{align*}
&\begin{pmatrix}
\frac{2}{9}(w^2)^2\d_x+\frac{2}{9}w^2w^2_x+\eps^2\left(\frac{5}{54}w^2\d_x^3+\frac{5}{36}w^2_x\d_x^2+\frac{1}{12}w^2_{xx}\d_x+\frac{1}{54}w^2_{xxx}\right)+\frac{\eps^4}{162}\d_x^5 & w^1\d_x+\frac{1}{3}w^1_x \\
 w^1\d_x+\frac{2}{3}w^1_x & \frac{2}{3}w^2\d_x+\frac{1}{3}w^2_x+\frac{2}{9}\eps^2\d_x^3
\end{pmatrix}.
\end{align*}
On the other hand, from~\cite[Proposition~3.7]{BG16}) and~\eqref{eq:properties of og} we obtain
\begin{gather*}
\og=\int\left(\frac{(u^1)^2 u^2}{2}+\frac{(u^2)^4}{72}+\eps^2\left(\frac{(u^2)^2 u^2_2}{144}+\frac{u^1 u^1_2}{24}\right)+\frac{\eps^4}{1728}u^2 u^2_4\right)dx.
\end{gather*}
Computing the operator $K_2$ using this expression (and identifying $u^\alpha=w^\alpha$) we get exactly the operator $K_2^{\text{$3$-spin}}$. 

For $r=4$ the operator $K^{\text{$4$-spin}}_2$ is given by
\begin{flalign*}
K^{\text{$4$-spin};1,1}_2=&\left(\frac{(w^3)^3}{32}+\frac{3(w^2)^2}{16}\right)\d_x+\frac{3w^2 w^2_x}{16}+\frac{3(w^3)^2 w^3_x}{64}&&\\
&+\eps^2\left(\left(\frac{7(w^3)^2}{256}+\frac{w^1}{48}\right)\d_x^3+\left(\frac{w^1_x}{32} +\frac{21w^3 w^3_x}{256}\right)\d_x^2+\left(\frac{5(w^3_x)^2}{128}+\frac{13w^3 w^3_{xx}}{256}+\frac{w^1_{xx}}{24}\right)\d_x\right.&&\\
&\hspace{1.2cm}\left.+\frac{3w^3_x w^3_{xx}}{128}+\frac{w^1_{xxx}}{64}+\frac{3w^3 w^3_{xxx}}{256}\right)&&\\
&+\eps^4\left(\frac{7w^3}{1152}\d_x^5 +\frac{35w^3_x}{2304}\d_x^4+\frac{91 w^3_{xx}}{4608}\d_x^3+\frac{133 w^3_{xxx}}{9216}\d_x^2+\frac{47 w^3_4}{9216}\d_x+\frac{w^3_5}{1536}\right)&&\\
&+\eps^6\frac{17}{36864}\d_x^7,&&
\end{flalign*}
\begin{flalign*}
K^{\text{$4$-spin};1,2}_2=&-\left(K^{\text{$4$-spin};2,1}_2\right)^\dagger&&\\
=&\frac{5w^2 w^3}{16}\d_x+\frac{w^3w^2_x}{8}+\frac{w^2 w^3_x}{8}+\eps^2\left(\frac{7w^2}{64}\d_x^3+\frac{7w^2_x}{48}\d_x^2+\frac{17w^2_{xx}}{192}\d_x+\frac{w^2_{xxx}}{48}\right),&&
\end{flalign*}
\begin{flalign*}
K^{\text{$4$-spin};1,3}_2=&-\left(K^{\text{$4$-spin};3,1}_2\right)^\dagger=w^1\d_x+\frac{w^1_x}{4}+\eps^2\left(\frac{7w^3}{192}\d_x^3+\frac{7w^3_x}{192}\d_x^2\right)+\eps^4\frac{7}{768}\d_x^5,&&
\end{flalign*}
\begin{flalign*}
K^{\text{$4$-spin};2,2}_2=&\left(\frac{(w^3)^2}{8}+w^1\right)\d_x+\frac{w^1_x}{2}+\frac{w^3w^3_x}{8}+\eps^2\left(\frac{w^3}{8}\d_x^3+\frac{3w^3_x}{16}\d_x^2+\frac{w^3_{xx}}{12}\d_x+\frac{w^3_{xxx}}{96}\right)+\frac{\eps^4}{64}\d_x^5,&&
\end{flalign*}
\begin{flalign*}
K^{\text{$4$-spin};2,3}_2=&-\left(K^{\text{$4$-spin};3,2}_2\right)^\dagger=\frac{3w^2}{4}\d_x+\frac{w^2_x}{4},&&
\end{flalign*}
\begin{flalign*}
K^{\text{$4$-spin};3,3}_2=&\frac{w^3}{2}\d_x+\frac{w^3_x}{4}+\eps^2\frac{5}{16}\d_x^3.&&
\end{flalign*}
On the other hand, from~\cite[Proposition~3.8]{BG16}) and~\eqref{eq:properties of og} we obtain
\begin{align*}
\og=&\int\left[\frac{(u^1)^2u^3}{2}+\frac{u^1(u^2)^2}{2}+\frac{(u^2)^2(u^3)^2}{16}+\frac{(u^3)^5}{960}+\right.\\
&\phantom{\int a}+\eps^2\left(\frac{u^1u^1_2}{16}+\frac{u^3_2(u^2)^2}{192}+\frac{u^3u^2u^2_2}{48}+\frac{u^1_2(u^3)^2}{192}+\frac{(u^3)^3u^3_2}{768}\right)+\\
&\phantom{\int a}\left.+\eps^4\left(\frac{u^2u^2_4}{640}+\frac{(u^3)^2u^3_4}{4096}+\frac{3u^1u^3_4}{2560}\right)+\eps^6\frac{u^3u^3_6}{49152}\right]dx.
\end{align*}
Computing the operator $K_2$ using this expression and performing the Miura transformation~\eqref{eq:Miura transformation for 4-spin} we obtain the operator $K_2^{\text{$4$-spin}}$. 

For $r=5$ the first part of Conjecture~\ref{main conjecture} is proved using the same scheme: we compute the operator $K_2^{\text{$5$-spin}}$ explicitly, then we compute the operator $K_2$ using the expression for the local functional~$\og_{1,1}$ from~\cite[Proposition~3.10]{BG16} and the first equation in~\eqref{eq:properties of og}, and finally we check that the resulting two operators are related by the Miura transformation~\eqref{eq:Miura transformation for 5-spin}. This computation is performed in \textsf{Mathematica} and is not presented here, since the computations and even the expressions for the Poisson operators get too involved.  

After we have proved the first part of the conjecture, the second part becomes easy in all three cases. Indeed, since $\mu_\alpha=\frac{2\alpha-r}{2r}$, one can easily check that Equations~\eqref{eq:DR biham} and~\eqref{eq:bihamiltonian recursion for DZ-rspin} agree.


\subsection{The Gromov--Witten theory of $\CP^1$}\label{subsubsection:projective line}

Consider the CohFT controlling the Gromov--Witten theory of $\CP^1$ (see e.g.~\cite[Section~6]{BR16}). For simplicity, we set the degree parameter (typically denoted by $q$) equal to $1$. This CohFT has rank $2$ and is homogeneous with
\begin{gather*}
(q_1,q_2)=(0,1),\quad (r^1,r^2)=(0,2),\quad \delta=1,\quad \eta=
\begin{pmatrix}
0 & 1\\
1 & 0
\end{pmatrix},\quad
\mu=
\begin{pmatrix}
-\frac{1}{2} & 0\\
0 & \frac{1}{2}
\end{pmatrix},\quad
A=
\begin{pmatrix}
2 & 0\\
0 & 2
\end{pmatrix}.
\end{gather*}
The corresponding DR hierarchy is studied in~\cite{BR16}. In particular, it is related to the extended Toda hierarchy (see e.g.~\cite{CDZ04}).

\begin{notation} In computations we use the following two formal power series:
$$
S(z)\coloneqq\frac{e^{z/2}-e^{-z/2}}{z} \qquad \text{and} \qquad \tS(z)\coloneqq\frac{e^{z/2}+e^{-z/2}}{2}.
$$ 
\end{notation}

Let us compute the operator $K_2$ for the DR hierarchy. 

\begin{lemma}
We have
\begin{gather}\label{eq:expression for og}
\og=\int\left(\frac{(u^1)^2u^2}{2}+\sum_{g\ge 1}\eps^{2g}\frac{B_{2g}}{(2g)(2g)!}u^1u^1_{2g}+e^{S(\eps\d_x)u^2}-u^2-\frac{(u^2)^2}{2}\right)dx,
\end{gather}
where $B_{2g}$ are the Bernoulli numbers.
\end{lemma}
\begin{proof}
By~\cite{BR16b} we have
$$
\og_{1,1}=\int\left(\frac{(u^1)^2u^2}{2}+\sum_{g\ge 1}\eps^{2g}\frac{B_{2g}}{(2g)!}u^1u^1_{2g}+\left(\tS(\eps\d_x)u^2-2\right)e^{S(\eps\d_x)u^2}+u^2\right)dx.
$$
Denote by $f$ the expression under the integral sign on the right-hand side of~\eqref{eq:expression for og}. Note that $f|_{\eps=0}\in\mbC[[u^1,u^2]]$ starts from cubic terms in the variables $u^1,u^2$. Therefore, it is sufficient to check that
$$
(D-2)f=\frac{(u^1)^2u^2}{2}+\sum_{g\ge 1}\eps^{2g}\frac{B_{2g}}{(2g)!}u^1u^1_{2g}+\left(\tS(\eps\d_x)u^2-2\right)e^{S(\eps\d_x)u^2}+u^2,
$$
where we recall that $D=\sum_{n\ge 0}(n+1)u^\alpha_n\frac{\d}{\d u^\alpha_n}$.
Since
\begin{align*}
(D-2)\frac{(u^1)^2u^2}{2}& =\frac{(u^1)^2u^2}{2}, 
& (D-2)u^1u^1_{2g} & =2gu^1u^1_{2g}, \\
(D-2)\frac{(u^2)^2}{2}&=0,
& (D-2)u^2 &=-u^2,
\end{align*}
it is sufficient to check that
$$
D e^{S(\eps\d_x)u^2}=\tS(\eps\d_x)u^2\cdot e^{S(\eps\d_x)u^2},
$$
which follows from
$$
D e^{S(\eps\d_x)u^2}=\left(\eps\d_x\circ S'(\eps\d_x)+S(\eps\d_x)\right)u^2\cdot e^{S(\eps\d_x)u^2}=\tS(\eps\d_x)u^2\cdot e^{S(\eps\d_x)u^2}.
$$
\end{proof}

We further compute
$$
\frac{\delta\og}{\delta u^1}=u^1u^2+\sum_{g\ge 1}\eps^{2g}\frac{B_{2g}}{g(2g)!}u^1_{2g},\qquad \frac{\delta\og}{\delta u^2}=\frac{(u^1)^2}{2}+S(\eps\d_x)e^{S(\eps\d_x)u^2}-1-u^2.
$$
Therefore,
\begin{align*}
\hOmega(\og)=&\begin{pmatrix}
S(\eps\d_x)\circ e^{S(\eps\d_x)u^2}S(\eps\d_x)-1 & u^1\\
u^1 & u^2+\sum_{g\ge 1}\frac{B_{2g}}{g(2g)!}(\eps\d_x)^{2g}
\end{pmatrix},\\
\hOmega^1(\og)=&\begin{pmatrix}
\eps\left(S(\eps\d_x)\circ e^{S(\eps\d_x)u^2}S'(\eps\d_x)+S'(\eps\d_x)\circ e^{S(\eps\d_x)u^2}S(\eps\d_x)\right) & 0\\
0 & \sum_{g\ge 1}\eps^{2g}\frac{2 B_{2g}}{(2g)!}\d_x^{2g-1}
\end{pmatrix}.
\end{align*}
Finally, we obtain
\begin{align*}
K_2=&
\begin{pmatrix}
S(\eps\d_x)\d_x\circ e^{S(\eps\d_x)u^2}\tS(\eps\d_x)+\tS(\eps\d_x)\circ e^{S(\eps\d_x)u^2}S(\eps\d_x)\d_x & u^1\d_x\\
\d_x\circ u^1 & \sum_{g\ge 0}\eps^{2g}\frac{2B_{2g}}{(2g)!}\d_x^{2g+1}
\end{pmatrix}\\
=&\begin{pmatrix}
\eps^{-1}\left(e^{\frac{e^{\eps\d_x}-1}{\eps\d_x}u^2}e^{\eps\d_x}-e^{\frac{1-e^{-\eps\d_x}}{\eps\d_x}u^2}e^{-\eps\d_x}\right) & u^1\d_x\\
\d_x\circ u^1 & \frac{e^{\eps\d_x}+1}{e^{\eps\d_x}-1}\eps\d_x^2
\end{pmatrix}.
\end{align*}

On the other hand, consider the extended Toda hierarchy from~\cite{CDZ04}.  The dependent variables of the extended Toda hierarchy are denoted by $v$ and $u$ in~\cite{CDZ04}.  Following~\cite{BR16} we denote $v^1=v$ and $v^2=u$. The extended Toda hierarchy is bihamiltonian with the pair of compatible Poisson operators $(K_1^\Td,K_2^{\Td})$ given by
\begin{align*}
K_1^\Td=&
\begin{pmatrix}
0 & \eps^{-1}\left(e^{\eps\d_x}-1\right)\\
\eps^{-1}\left(1-e^{-\eps\d_x}\right) & 0
\end{pmatrix},\\
K_2^\Td=&
\begin{pmatrix}
\eps^{-1}\left(e^{\eps\d_x}\circ e^{v^2}-e^{v^2}e^{-\eps\d_x}\right) & \eps^{-1}v^1\left(e^{\eps\d_x}-1\right)\\
\eps^{-1}\left(1-e^{-\eps\d_x}\right)\circ v^1 & \eps^{-1}\left(e^{\eps\d_x}-e^{-\eps\d_x}\right)
\end{pmatrix}.
\end{align*}
It is straightforward to check that the Miura transformation
$$
u^1=e^{-\frac{\eps\d_x}{2}}v^1,\qquad u^2=\frac{\eps\d_x}{e^{\frac{\eps\d_x}{2}}-e^{-\frac{\eps\d_x}{2}}}v^2
$$
takes the pair $(K_1^\Td,K_2^{\Td})$ to the pair $(K_1,K_2)$. Thus, Part~1 of Conjecture~\ref{main conjecture} holds for the CohFT of the Gromov--Witten theory of $\CP^1$.

Let us check Part~2 of Conjecture~\ref{main conjecture}. Let $\oh^\Td_{\alpha,d}$ denote the Hamiltonians of the extended Toda hierarchy. In~\cite{BR16} it is proved that
\begin{gather}\label{eq:DR and Toda Hamiltonians}
\og_{\alpha,d}=\sum_{i=0}^{d+1}(-1)^i(S_i^*)^\mu_\alpha\oh^\Td_{\mu,d-i}[u],\quad d\ge -1,
\end{gather}
where $S_i^*$, $i\ge 0$, are certain matrices described in~\cite[Section~6]{BR16}. The bihamiltonian recursion for the extended Toda hierarchy is given by (see \cite[Theorem~3.1]{CDZ04})
$$
\left\{\cdot,\oh_{\alpha,d}^\Td\right\}_{K_2^\Td}=\left(d+\frac{3}{2}+\mu_\alpha\right)\left\{\cdot,\oh_{\alpha,d+1}^\Td\right\}_{K_1^\Td}+R^\beta_\alpha\left\{\cdot,\oh_{\beta,d}^\Td\right\}_{K_1^\Td},\quad d\ge -1,
$$
where
$$
R=(R^\beta_\alpha)\coloneqq\begin{pmatrix}
0 & 0\\
2 & 0
\end{pmatrix}.
$$
Using~\eqref{eq:DR and Toda Hamiltonians} and the property
$$
A_\alpha^\gamma(S_i^*)^\beta_\gamma=(S_i^*)_\alpha^\gamma R^\beta_\gamma+(S_{i+1}^*)^\beta_\alpha(i+1+\mu_\alpha-\mu_\beta),\quad i\ge 0,
$$
which is a standard fact in Gromov--Witten theory (or can be checked using the explicit formula for the matrices $S_i^*$ given in~\cite[Section 6]{BR16}), we obtain~\eqref{eq:DR biham} in this case.

\end{document}